\documentclass[a4paper]{article}

\usepackage{mathtools}
\usepackage[utf8]{inputenc}
\usepackage{amssymb,amsmath,enumerate}
\usepackage{fancyhdr}
\usepackage{stmaryrd}
\usepackage{xspace}
\usepackage{tikz}
\usepackage{multicol}
\usepackage{amsthm}
\usepackage{amsfonts}
\usepackage{amssymb}
\usepackage{amsmath}
 \usepackage{microtype}

\usepackage{tabularx,hhline,rotating,enumitem,xspace}
\usepackage[hidelinks, final]{hyperref}
\usepackage{tikz}
\usetikzlibrary{calc}

\usepackage{mathdots}

\usepackage{prettyref}

\newtheorem{theorem}{Theorem}
\newtheorem{lemma}[theorem]{Lemma}
\newtheorem{proposition}[theorem]{Proposition}
\newtheorem{corollary}[theorem]{Corollary}
\newtheorem{lemma*}{Lemma}
\newtheorem*{claim*}{Claim}

\theoremstyle{definition}
\newtheorem{definition}[theorem]{Definition}

\newtheorem{example}[theorem]{Example}
\newtheorem{algorithm}[theorem]{Algorithm}

\newcounter{AlgorithmusCounter}[section]

\renewcommand{\theAlgorithmusCounter}{\arabic{chapter}.\arabic{AlgorithmusCounter}}

\renewcommand{\setminus}{\mysetminus}

\newenvironment{aw}{\noindent\color{magenta} AW : }{}

\makeindex

\newcommand{\prref}\prettyref
\newrefformat{thm}{Theorem~\ref{#1}}
\newrefformat{lem}{Lem\-ma~\ref{#1}}
\newrefformat{def}{Definition~\ref{#1}}
\newrefformat{cor}{Corollary~\ref{#1}} 
\newrefformat{prop}{Proposition~\ref{#1}}
\newrefformat{sec}{Section~\ref{#1}}
\newrefformat{cap}{Chapter~\ref{#1}}
\newrefformat{bsp}{Example~\ref{#1}}
\newrefformat{rem}{Remark~\ref{#1}}
\newrefformat{fig}{Figure~\ref{#1}}
\newrefformat{eq}{(\ref{#1})}
\newrefformat{ex}{Example~\ref{#1}}
\newrefformat{tab}{Table~\ref{#1}}
\newrefformat{app}{Appendix~\ref{#1}}
\newrefformat{alg}{Algorithm~\ref{#1}}

\newcommand{\mysetminusD}{\raisebox{.8pt}{\hbox{\tikz{\draw[line width=0.6pt,line cap=round] (3.5pt,0pt) -- (0,5.2pt);}}}}
\newcommand{\mysetminusT}{\mysetminusD}
\newcommand{\mysetminusS}{\raisebox{.5pt}{\hbox{\tikz{\draw[line width=0.45pt,line cap=round] (2.2pt,0) -- (0,3.8pt);}}}}
\newcommand{\mysetminusSS}{\raisebox{.35pt}{\hbox{\tikz{\draw[line width=0.4pt,line cap=round] (1.5pt,0) -- (0,2.8pt);}}}}

\newcommand{\mysetminus}{\mathbin{\mathchoice{\mysetminusD}{\mysetminusT}{\mysetminusS}{\mysetminusSS}}}

\newcommand{\eqF}{\mathrel{=_{F(\cG)}}}

\newcommand{\set}[2]{\left\{#1\; \middle|\; #2\right\}}

\newcommand{\oneset}[1]{\left\{\mathinner{#1}\right\}}

\newcommand{\abs}[1]{\left|\mathinner{#1}\right|}

\newcommand{\floor}[1]{\left\lfloor\mathinner{#1} \right\rfloor}

\newcommand{\gen}[1]{\left< \mathinner{#1} \right>}
\newcommand{\genr}[2]{\left< \, \mathinner{#1}\; \middle|\;\mathinner{#2} \, \right>}

\newcommand{\N}{\ensuremath{\mathbb{N}}}
\newcommand{\Z}{\ensuremath{\mathbb{Z}}}
\newcommand{\Q}{\ensuremath{\mathbb{Q}}}

\newcommand{\EXPSPACE}{\ensuremath{\mathsf{EXPSPACE}}\xspace} 
\newcommand{\LOGDCFL}{\ensuremath{\mathsf{LOGDCFL}}\xspace} %
\renewcommand{\L}{\ensuremath{\mathsf{LOGSPACE}}\xspace} %
\newcommand{\TC}{\ensuremath{\mathsf{uTC}^0}\xspace}
\newcommand{\Tc}[1]{\ensuremath{\mathsf{TC}^{#1}}\xspace}
\newcommand{\Ac}[1]{\ensuremath{\mathsf{AC}^{#1}}\xspace}

\newcommand{\AC}{\ensuremath{\mathsf{uAC}^0}\xspace}
\newcommand{\NC}{\ensuremath{\mathsf{NC}}\xspace}
\renewcommand{\P}{\ensuremath{\mathsf{P}}\xspace}

\renewcommand{\phi}{\varphi}

\newcommand{\alp}{\alpha}
\newcommand{\bet}{\beta}

\newcommand{\lam}{\lambda}

\newcommand{\Sig}{\Sigma}

\newcommand{\Del}{\Delta}

\newcommand\GG{\Gamma}

\newcommand{\Oh}{\mathcal{O}}

\newcommand{\cB}{\mathcal{B}}

\newcommand{\cC}{\mathcal{C}}

\newcommand{\cG}{\mathcal{G}}

\newcommand{\cF}{\mathcal{F}}

\newcommand{\cP}{\mathcal{P}}

\newcommand{\BS}[2]{\ensuremath{\mathrm{\bf{BS}}_{#1,#2}}\xspace}

\newcommand{\BG}{\ensuremath{\mathrm{\bf{G}}_{1,2}}\xspace 
}

\newcommand\ra{\longrightarrow}

\newcommand\RAS[2]{\overset{#1}{\underset{#2}{\Longrightarrow}}}

\newcommand\DAS[2]{\overset{#1}{\underset{#2}{\Longleftrightarrow}}}

\newcommand\RAA{\Longrightarrow}

\newcommand\RAO[1]{\overset{#1}{\Longrightarrow}}

\newcommand\DAO[1]{\overset{#1}{\Longleftrightarrow}}

\newcommand{\smalloverline}[1]
{{\mspace{.8mu}\overline{\mspace{-.8mu}#1\mspace{-.8mu}}\mspace{.8mu}}}
\newcommand{\ov}[1]{\smalloverline{#1}}
\newcommand{\oi}[1]{{#1}^{-1}}

\newcommand{\rquot}[2]{{#1} /\!\;\!{#2}}
\newcommand{\rQuot}[2]{{#1} / {#2}}

\newcommand{\tto}{\mathrel{\tilde\to}}

\newcommand{\sse}{\subseteq}

\newcommand{\sm}{\setminus}

\newcommand{\proc}[1]{\textsc{#1}}

\newcommand\ie{i.\,e., }
\newcommand\Wlog{W.\,l.\,o.\,g.\ }

\newcommand\eg{e.\,g.\xspace}

\newcommand{\tar}{\tau}
\newcommand{\sor}{\iota}

\newcommand{\y}{y}

\tikzset{
	ncbar angle/.initial=90,
	ncbar/.style={
		to path=(\tikztostart)
		-- ($(\tikztostart)!#1!\pgfkeysvalueof{/tikz/ncbar angle}:(\tikztotarget)$)
		-- ($(\tikztotarget)!($(\tikztostart)!#1!\pgfkeysvalueof{/tikz/ncbar angle}:(\tikztotarget)$)!\pgfkeysvalueof{/tikz/ncbar angle}:(\tikztostart)$)
		-- (\tikztotarget)
	},
	ncbar/.default=0.5cm,
}

\begin{document}
	\title{A Logspace Solution to the Word and Conjugacy Problem of Generalized Baumslag-Solitar Groups} 
		
 \author{
 	Armin Wei\ss \medskip\\
 	\normalsize FMI, Universit{\"a}t Stuttgart, Germany
 }
 
 \date{\today}

\maketitle

\begin{abstract}
	Baumslag-Solitar groups were introduced in 1962 by Baumslag and Solitar as examples for finitely presented non-Hopfian two-generator groups. Since then, they served as examples for a wide range of purposes. As Baumslag-Solitar groups are HNN extensions, there is a natural generalization in terms of graph of groups. 
	
	Concerning algorithmic aspects of generalized Baumslag-Solitar groups, several decidability results are known. Indeed, a straightforward application of standard algorithms leads to a polynomial time solution of the word problem (the question whether some word over the generators represents the identity of the group). The conjugacy problem (the question whether two given words represent conjugate group elements) is more complicated; still decidability has been established by Anshel and Stebe for ordinary Baumslag-Solitar groups and for generalized Baumslag-Solitar groups independently by Lockhart and Beeker. 
	However, up to now, no precise complexity estimates have been given.
	
	In this work, we give a \L algorithm for both problems. More precisely, we describe a uniform \Tc{0} many-one reduction of the word problem to the word problem of the free group. Then we refine the known techniques for the conjugacy problem and show it is \Ac{0}-Turing-reducible to the word problem of the free group.
	
	Finally, we consider uniform versions (where also the graph of groups is part of the input) of both word and conjugacy problem: while the word problem still is solvable in \L, the conjugacy problem becomes \EXPSPACE-complete.
	
	\smallskip
	
	\paragraph{Keywords:} 
	word problem,
	conjugacy problem,
	Baumslag-Solitar group,
	graph of groups,
	Logspace\\
\end{abstract}

\section{Introduction}\label{sec:intro} 

A \emph{Baumslag-Solitar group} is a group of the form $\BS{p}q = \genr{a,y}{ya^p\oi y = a^q}$ for some $p,q \in \Z\setminus \oneset{0}$. These groups were introduced in 1962 by Baumslag and Solitar \cite{baumslag62some} as examples for finitely presented non-Hopfian two-generator groups. They showed that the class of Baumslag-Solitar groups comprises both Hopfian and non-Hopfian groups.

The usual presentation of a Baumslag-Solitar groups is as HNN extension of an infinite cyclic group with one stable letter. The different Baumslag-Solitar groups correspond to the different inclusions of the associated subgroup into the base group.
HNN extensions are a special case of fundamental groups of a graph of groups~-- where the graph consists of exactly one vertex with one attached loop.
Thus, there is a natural notion of generalized Baumslag-Solitar group (GBS group) as fundamental group of a graph of groups with infinite cyclic vertex and edge groups~-- see \eg\ \cite{Beeker11thesis, Forester03}. GBS groups were also studied in \cite{Kropholler90} and characterized as those finitely presented groups of cohomological dimension two which have an infinite cyclic subgroup whose commensurator is the whole group.

Algorithmic problems in group theory were introduced by Max Dehn more than 100 years ago. The two basic problems are the word problem and the conjugacy problem, which are defined as follows: Let $G$ be a finitely generated group.\\ 
\emph{Word problem}: On input of some word $w$ written over the generators, decide whether $w = 1$ in $G$. \\
\emph{Conjugacy problem}: On input of two words $v$ and $w$ written over the generators, decide whether $v $ and $w$ are conjugate, \ie whether there exists $z\in G$ such that $zv z^{-1}= w$ in $G$.

In recent years, conjugacy played an increasingly important role in non-commuta\-tive cryptography, see \eg~\cite{CravenJ12,GrigorievS09,SZ1}. These applications use that it is easy to create elements which are conjugated, but to check whether two given elements are conjugated might be difficult~-- even if the word problem is easy. In fact, there are groups where the word problem is easy but the conjugacy problem is undecidable \cite{Miller1}. 

It has been long known that both the word problem and the conjugacy problem in generalized Baumslag-Solitar groups are decidable. Actually, the standard application of Britton reductions leads to a polynomial time algorithm for the word problem (see \eg\ \cite{Laun12diss}). Decidability of the conjugacy problem has been shown by Anshel and Stebe for ordinary Baumslag-Solitar groups \cite{AnshelS74} and for arbitrary GBS groups independently by Lockhart \cite{Lockhart92} and Beeker \cite{Beeker11thesis}. 

The probably first non-trivial complexity bounds for the word problem have been established by the general theorem by Lipton and Zalcstein \cite{lz77} resp.\ Simon \cite{Sim79} that linear groups have word problem in \L (although linear GBS groups form a small sub-class of all GBS groups).
Later, Waack \cite{Waack81} examined the particular GBS group $\genr{a,s,t}{sa \oi s = a, ta\oi t = a^2}$ as an example of a non-linear group which has word problem in \L. In order to obtain the \L bound for the word problem, he used the very special structure of this particular GBS group: the kernel under the canonical map onto the solvable Baumslag-Solitar group $\BS{1}{2}$ is a free group.

For solvable GBS groups~-- which are precisely the Baumslag-Solitar groups $\BS{1}{q}$ for $q \in \Z$~-- the word problem was shown to be in (non-uniform) \Tc{0} by Robinson \cite{Robinson93phd} (and also in \L). 
Moreover, in \cite{DiekertMW14} it is shown that both the word and the conjugacy problem in \BS12 is in uniform \Tc{0}, indeed. It is straightforward to see that this proof also works for $\BS{1}{q}$ for arbitrary $q$, see \cite{Weiss15diss}. The result for the conjugacy problem became possible because of the seminal theorem by Hesse \cite{hesse01, HeAlBa02} that integer division is in uniform \Tc{0}~-- a result which also plays a crucial role in this work.

Apart from these (and some other) special cases, no precise general complexity estimates have been given. In this work, we show that both the word problem and the conjugacy problem of every generalized Baumslag-Solitar group is in \L.
More precisely, we establish the following results:
{\renewcommand*{\thetheorem}{\Alph{theorem}}
	\begin{theorem}\label{thm:wp}
		Let $G$ be a GBS group. There is a uniform \Tc0 many-one reduction from the word problem of $G$ to the word problem of the free group $F_2$.
	\end{theorem}
	
	Together with the well-known result that 
	linear groups~-- in particular $F_2$~-- have word problem in $\L$ \cite{lz77,Sim79}, this leads to a \L algorithm of the word problem. Moreover, in view of \cite{CaussinusMTV98}, \prref{thm:wp} shows that the word problem of GBS groups is in the complexity class $\mathsf{C_=}\NC^1$ (for a definition see \cite{CaussinusMTV98}).
	
	\begin{theorem}\label{thm:cp}
		Let $G$ be a GBS group. The conjugacy problem of $G$ 
		is uniform-$\Ac0$-Turing-reducible to the word problem of the free group.
	\end{theorem}	
	
	We also consider uniform versions of the word and conjugacy problem (where the GBS group is part of the input~-- for precise definitions see \prref{sec:uniWP} and \prref{sec:uniCP}). This leads to the following contrasting theorem:
		
	\begin{theorem}\label{thm:uni}\
		\begin{enumerate}
			\item The uniform word problem for GBS groups is in \L. Moreover, if the GBS groups are given as fundamental groups with respect to a spanning tree, the uniform word problem is \L-complete.
			\item The uniform conjugacy problem for GBS groups is \EXPSPACE-complete.
		\end{enumerate}
	\end{theorem}
	\setcounter{theorem}{0}
}

The paper is organized as follows: in \prref{sec:prelims}, we fix our notation and recall some basic facts on complexity and graphs of groups~-- the reader who is familiar with these concepts might skip that section and only consult it for clarification. In \prref{sec:wp}, we give the proof of \prref{thm:wp}, describe how to compute Britton-reduced words, and consider the uniform word problem. Finally, \prref{sec:cp} deals with the non-uniform and uniform version of the conjugacy problem.
Parts of this work are also part of the author's dissertation \cite{Weiss15diss}.

\section{Preliminaries}\label{sec:prelims} 
\setcounter{subsection}{1}
\paragraph{Words.} An \emph{alphabet} is a (finite or infinite) set $\Sig$; an element $a \in \Sig$ is called a \emph{letter}. 
The free monoid over $\Sig$ is denoted by $\Sig^*$, its elements
are called {\em words}. The multiplication of the monoid is concatenation of words. The identity element is the empty word $1$.  
If $w,p,x,q$ are words with $w = pxq$, then we call $x$ a \emph{factor} of $w$. 

\paragraph{Rewriting systems.}
Let $X$ be a set; a \emph{rewriting system} over $X$ is a binary relation
$\RAA\,\subseteq\, X \times X$. If $(x,y)\in\, \RAA$, we write $x \RAA y$.
The idea of the notation is that $x\RAA y$ indicates that $x$ can be rewritten into $y$ in one step. 
We denote the reflexive and transitive closure of $\RAA$ by $\RAO*$; and by $\DAO*$ its reflexive, transitive, and symmetric closure~-- it is the smallest equivalence relation
such that $x$ and $y$ are in the same class for all $x \RAA y$.

\paragraph{Rewriting over words.}\label{sec:present} 
\newcommand\SGr{\Sig^*}
~Let $\Sig$ be an alphabet and $S \sse \SGr \times \SGr$ be a set of pairs. This defines a rewriting system $\smash{\RAS{}S} \vphantom{\DAS{}{}}$ over $\SGr$ by $\smash{x \RAS{}S y}$ if $x = u\ell v$ and $y = ur v$ for some $(\ell,r)\in S$. 
It is common to denote a rule $(\ell,r)\in S$ by $\ell \to r$ and we call $S$ itself a rewriting system. Since $\DAS*{S}$ is an equivalence relation, we can form the set of equivalence classes $\smash{\rquot{\SGr}{S} = \set{[x]}{x \in \SGr}}$, where $[x] = \set{y \in \SGr}{\smash{x \DAS*{S} y}} \vphantom{\DAS{}{}}$. 
Now, $\rquot{\SGr}{S}$ becomes a monoid by $ [x]\cdot [y] = [xy]$, and the mapping $x\to [x]$ yields a canonical homomorphism $\eta:\SGr\to \rquot{\SGr}{S}$. 	
%
The rewriting system $S$ is called 
\begin{itemize}
	\item \emph{confluent} if $x\RAS *{S}y$ and $x\RAS *{S}z$ implies
	$\smash{\exists \, w \, :\, \smash{y\RAS *{S}w} \text{ and }\smash{z\RAS *{S}w}}$,
	\item \emph{terminating} if there are no infinite chains 
	$	x_0 \RAS {}{S} x_1 \RAS {}{S} x_2 \RAS {}{S} \cdots,$
\end{itemize}
A rewriting system $S$ is confluent if and only if $x\DAS *{S}y$ implies
$\exists \, w \, :\, x\RAS *{S}~w$ and $y\RAS *{S}w$ (see \cite{bo93springer,jan88eatcs}).
Thus, if $S$ is confluent and teminating, then in every class of $\rquot{\Sig^*}{S}$ there is exactly one element to which no rule of $S$ can be applied.

\paragraph{Groups.}
We consider a group $G$ together with a surjective homomorphism $\eta:\Sig^* \to G$ (a \emph{monoid presentation}) for some (finite or infinite) alphabet $\Sig$. 
In order to keep notation simple, we suppress the homomorphism $\eta$ and consider words also as group elements. We write $w=_{G}w'$ as a shorthand of 
$\eta(w)=\eta(w')$ and $ w \in_G A$ instead of $\eta(w) \in A$ for $A \sse G$ and $w \in \Sig^*$. 

For words (or group elements) $v,w$ we write $v \sim_G w$ to denote conjugacy, \ie $v \sim_G w$ if and only if there exists some $z\in G$ such that $zv \oi z =_G w$. If $H$ is a subgroup of $G$, we write $v \sim_H w$ if there is some $z\in H$ such that $zv \oi z =_G w$.

\paragraph{Involutions.} An involution on a set $\Sig$ is a mapping 
$x \mapsto \ov x$ such that $\ov {\ov x} = x$. 
We consider only fixed-point-free involutions, \ie $x \neq \ov x$.

\paragraph{Free groups.} Let $\Lambda$ be some alphabet and set $\Sig = \Lambda \cup \ov \Lambda$ where $\ov \Lambda = \set{\ov a}{a \in \Lambda}$ is a disjoint copy of $\Lambda$. There is a fixed-point-free involution $\ov{\,\cdot\,}: \Sig \to \Sig$ defined by $a \mapsto \ov a$ and $\ov a \mapsto a$ (\ie $\ov{\ov a} = a$). 
Consider the confluent and terminating rewriting system of \emph{free reductions} $S = \set{a\ov a\to 1 } { a\in \Sig}$. Some word $w\in \Sig^*$ is called \emph{freely reduced} if there is no factor $a\ov a$ for any letter $a \in \Sig$.
The rewriting system $S$ defines the \emph{free group} $F_\Lambda = \rquot{\Sig^*}{S}$. We have $\ov a =_{F_\Lambda} \oi a$ for $a \in \Sig$.
We write $F_2$ as shorthand of $F_{\oneset{a,b}}$.

\paragraph{Graphs.} 
For the notation of graphs we follow Serre's book \cite{Serre77}. 
A \emph{graph} $Y = (V,E,\sor,\tar,\ov{\,\cdot\,})$ is given by the following data: 
a set of \emph{vertices} $V= V(Y)$ and a set of \emph{edges} $E=E(Y)$ together with two mappings $\sor,\tar:E \to V$
and an involution $e \mapsto \ov e$ without fixed points such that $\sor(e) = \tar(\ov e)$. 

An \emph{orientation} of a graph $Y$ is a subset $D \sse E$ such that $E$ is the disjoint union $ E= D \cup \ov{D}$.
A \emph{path} with start point $u$ and end point $v$ is a sequence of edges $e_1, \dots, e_n$ such that $\tar(e_i) = \sor(e_{i+1})$ for all $i$ and $\sor(e_1) = u$ and $\tar(e_n)= v$. A graph is \emph{connected} if for every pair of vertices there is a path connecting them.

\subsection{Complexity}
Computation or decision problems are given by functions $f:\Del^* \to \Sig^*$ for some finite alphabets $\Del $ and $\Sig$. In case of a decision problem (or formal language) the range of $f$ is the two element set $\oneset{0,1}$.

\smallskip\emph{\L} is the class of functions computable by a deterministic Turing machine with working tape bounded logarithmically in the length of the input.

Our result uses the following well-known theorem about linear groups (groups which can be embedded into a matrix group over some field). It was obtained by Lipton and Zalcstein \cite{lz77} for fields of characteristic 0 and by Simon \cite{Sim79} for other fields.
\begin{theorem}[\cite{lz77,Sim79}]\label{thm:linlog}
	Linear groups have word problem in $\L$.
\end{theorem}

\paragraph{Circuit Complexity.}
The class \Ac{0} (resp.\ \Tc{0}) is defined as the class of functions computed by families of circuits of constant depth and polynomial size with unbounded fan-in Boolean gates (and, or, not) (resp.\ unbounded fan-in Boolean and \proc{Majority} gates)~-- the alphabets $\Del$ and $\Sig$ are encoded over the binary alphabet $\oneset{0,1}$.
In the following, we only consider $\mathsf{Dlogtime}$-uniform circuit families and we write $\AC$ (resp.\ $\TC$) as shorthand for $\mathsf{Dlogtime}$-uniform $\Ac{0}$ (resp.\ $\Tc0$). $\mathsf{Dlogtime}$-uniform means that there is a deterministic Turing machine which decides in time $\Oh(\log n)$ on input of two gate numbers (given in binary) and the string $1^n$ whether there is a wire between the two gates in the $n$-input circuit and also decides of which type some gates is. Note that the binary encoding of the gate numbers requires only $\Oh(\log n)$ bits~-- thus, the Turing machine is allowed to use time linear in the length of the encodings of the gates.
For more details on these definitions we refer to \cite{Vollmer99}.

\paragraph{Reductions.} Let $K\sse \Del^*$ and $L \sse \Sig^*$ be languages and $\cC$ a complexity class. Then $K$ is called $\cC$-many-one-reducible to $L$ if there is a $\cC$-computable function $f: \Del^* \to \Sig^*$ such that $w \in K$ if and only if $f(w) \in L$.

A function $f$ is \emph{\AC-reducible} (or \AC-Turing-reducible) to a function $g$ if there is a $\mathsf{Dlogtime}$-uniform family of \Ac0 circuits computing $f$ which, in addition to the Boolean gates, also may use oracle gates for $g$ (\ie gates which on input $x$ output $g(x)$). 
We write $\AC(F_2)$ for the family of problems which are \AC-reducible to the word problem of the free group $F_2$.

\paragraph{The Class \TC and Arithmetic.} 
Although \TC is a very low parallel complexity class, it is still very powerful with respect to arithmetic. 
By the very definition of \AC reducibility, 
\proc{Majority} is \TC-complete. As an immediate consequence, the word problem of $\Z$ with generators $\pm 1$ is also \TC-complete (since a sequence over the alphabet $\oneset{\pm 1}$ sums up to $0$ if and only if there is neither a majority of letters $1$ nor of letters $-1$).

\proc{Iterated Addition} (resp.\ \proc{Iterated Multiplication}) are the following computation problems: On input of $n$ binary integers $a_1,\dots, a_n$ each having $n$ bits (\ie the input length is $N=n^2$), compute the binary representation of the sum $\sum_{i=0}^n a_i$ (resp.\ product $\prod_{i=0}^n a_i$).
For \proc{Integer Division},\index{Integer Division@\proc{Integer Division}} the input are two binary $n$-bit integers $a, b$; the binary representation of the integer $c=\floor{a/b}$ has to be computed.
The first statement of \prettyref{thm:divisionTC} is a standard fact, see \cite{Vollmer99}; the other statements are due to Hesse, \cite{hesse01, HeAlBa02}.

\begin{theorem}[\cite{hesse01, HeAlBa02, Vollmer99}]\label{thm:divisionTC}\label{thm:additionTC}
	The problems \proc{Iterated Addition},
	\proc{Iterated Multiplication},
	\proc{Integer Division} are all in \TC.
\end{theorem}
We have the following inclusions (note that even $\TC \sse \P$ is not known to be strict):
\begin{align*}
\TC \sse \AC(F_2) \sse \L \sse \P.
\end{align*}
The first inclusion is because there is a subgroup $\Z$ in $F_2$; the second inclusion is because of \prref{thm:linlog}.

\newcommand{\cu}{c}
\newcommand{\PP}{a}
\newcommand{\QQ}{b}

\subsection{Graphs of Groups}\label{cap:bass_serre}
\newcommand{\BSalp}{\Del}

Since generalized Baumslag-Solitar groups are defined as fundamental groups of graphs of groups, we give a brief introduction into this topic. 
Our presentation is a shortened version taken from \cite{DiekertW13cfBST}, which in turn is based on Serre's book \cite{Serre77}.

\begin{definition}[Graph of Groups]
	Let $Y = (V(Y),E(Y))$ be a connected graph. 
	A \emph{graph of groups} $\cG$ over $Y$ is given by the following data:
	\begin{enumerate}
		\item For each vertex $\PP \in V(Y)$, there is a \emph{vertex group} 
		$G_\PP$.
		\item For each edge $y \in E(Y) $, there is an \emph{edge group}
		$G_y $ such that $G_y = G_{\ov y}$.
		\item For each edge $y \in E(Y) $, there is
		an injective homomorphism from $G_y$ to $G_{\sor(y)}$, which is denoted by $\cu \mapsto \cu^y$.
		The image of $G_y$ in $G_{\sor(y)}$ is denoted by $G_y^y$.
	\end{enumerate}
	In the following, $Y$ is always a finite graph.	Since $G_y = G_{\ov y}$, there is also a homomorphism $G_y \to G_{\tar(y)}$. 
	Thus, for $y \in E(Y) $ with $\sor(y)= \PP$ and $\tar(y)= \QQ$, there are two 
	isomorphisms and inclusions: 
	\begin{align*}
		G_y \tto G_y^y \leq G_\PP &,\quad \cu \mapsto \cu^y, &
		G_y \tto G_{\ov y}^{\ov y} \leq G_\QQ &, \quad \cu \mapsto \cu^{\ov y}.
	\end{align*} 
\end{definition}

The fundamental group of $\cG$ can be constructed as subgroup of the larger group $F(\cG)$:
as an (possibly infinite) alphabet we choose a disjoint union 
\[\BSalp= E(Y) \cup\!\!\! {\bigcup_{\PP\in V(Y)} \left(G_\PP \setminus \oneset{1}\right),}\]
and we define the group
\[F(\cG) =\rQuot{\Del^*}{\set{ gh=[gh],\, y\cu^{\ov{y}}\ov{y}=\cu^y}{\PP\in V(Y),\, g,h \in G_\PP;\, y\in E(Y), \, \cu\in G_y},}\]
where $[gh]$ denotes the element obtained by multiplying $g$ and $h$ in $G_\PP$ (where $1 \in G_\PP$ is identified with the empty word). 

Let us define subsets of $\Delta^*$ as follows: for $\PP,\QQ\in V(Y) $, we denote with ${\Pi}(\cG, \PP,\QQ)$ the set of words where the occurring edges form a path from $\PP$ to $\QQ$ in $Y$ and the elements of vertex groups between two edges are from the corresponding vertex in the path; more precisely, 
\begin{align*}
	\begin{split}
		{\Pi}(\cG, \PP,\QQ) = \big\{\,g_0y_1\cdots g_{n-1}y_n g_{n} \;\big\vert\; y_i \in E(Y),\: \sor(y_1) = \PP,\: \tar(y_n) = \QQ,\,\qquad\qquad \\ 
		\tar(y_{i}) = \sor(y_{i+1}),\, g_{0} \in G_{\PP},\,g_i\in G_{\tar(y_{i})} \text{ for all $i $}\,\big\},
	\end{split}
\end{align*}
where again $1 \in G_\PP$ is identified with the empty word. Moreover, we set 
$${\Pi}(\cG) = \bigcup_{\PP\in V(Y)} {\Pi}(\cG, \PP,\PP).$$

In general, the image of $\Pi(\cG)$ in $F(\cG)$ is not a group but a so-called \emph{groupoid}.
If $w =g_0y_1\cdots g_{n-1}y_n g_{n} \in {\Pi}(\cG) $, then we call $w$ a \emph{$\cG$-factorization} of the respective group element in $F(\cG)$; by saying this we implicitly require that $ y_i \in E(Y),\, \tar(y_{i}) = \sor(y_{i+1}),\,g_i\in G_{\tar(y_{i})}$ for all $i$, $\tar(y_{n}) = \sor(y_{1})$, and $g_{0} \in G_{\sor(y_1)}$.
We call $ y_1\cdots y_n$ the \emph{underlying path} of $w$. 

For all vertices $\PP\in V(Y)$, the image of $\Pi(\cG,\PP,\PP)$ in $F(\cG)$ is a group.
\begin{definition}\ \label{def:fgogP}
	\begin{enumerate} 
		\item Let $\PP\in V(Y)$. The \emph{fundamental group} $\pi_1(\cG,\PP)$ of $\cG$ with respect to the base point $\PP \in V(Y)$ is defined as the image of $\Pi(\cG,\PP,\PP)$ in $F(\cG)$.
		\item Let $T$ be a spanning tree of $Y$ (\ie a subset of $E(Y)$ connecting all vertices and not containing any cycles). The \emph{fundamental group} of $\cG $ with respect to $T$ is defined by
		\[\pi_1(\cG,T)= \rQuot{F(\cG)}{\set{y=1}{y\in T}.}\]
	\end{enumerate}
\end{definition}

\begin{proposition}[\cite{Serre77}]\label{prop:twofunds}
	The canonical homomorphism from the subgroup $\pi_1(\cG, \PP)$ 
	of $F(\cG)$ to the quotient group $\pi_1(\cG,T)$ is an isomorphism. In particular, the two definitions of the fundamental group are independent of the choice of the base point and the spanning tree.
\end{proposition}

\begin{example}\label{ex:bs}
	Let $\cG$ be a graph of groups over the following graph:
	
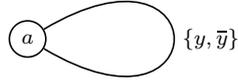
\vspace*{-0.7cm}
	\begin{center}
		\begin{footnotesize}
			\begin{tikzpicture}
			\node [draw,semithick,circle ](b) at (0,0 ) {$\PP$};
			\draw [semithick,-] (b) .. controls(2.5,1.5) and (2.5,-1.5) .. (b)node[ pos=.5,right]{$\oneset{y,\ov y}$};
			\end{tikzpicture}
		\end{footnotesize}
	\end{center}
	\vspace*{-1cm}
	and let $G_\PP= \Z = \gen{a}$ and $G_y=G_{\ov y} = \Z = \gen{c}$ and the inclusions given by 
	\begin{align*}
		c&\mapsto a^p & c\mapsto a^q
	\end{align*}
	for some $p,q \in \Z \setminus \oneset{0}$.
	Then the fundamental group ${\pi_1}(\cG,\PP)$ is the 
	Baumslag-Solitar group
	$${\pi_1}(\cG,\PP) = \BS{p}{q} = \genr{a,y}{ya^p\oi y = a^q}.$$
\end{example}

\paragraph{Britton Reductions over Graphs of Groups.}\label{sec:britton}
In \cite{britton63}, \emph{Britton reductions} were originally defined for HNN extensions. They are given by the rewriting system $B_\cG\sse \BSalp^*\times \BSalp^*$ with the following rules (see also \cite[Sec.\ IV.2]{LS01}):
\begin{align*}
	&&gh &\ra{} [gh] && \quad \text{for } \PP\in V(Y),\ g, h \in G_\PP \sm \oneset{1},\\
	&&y \cu^{\ov{y}}\ov{y}&\ra{} \cu^y && \quad \text{for } y\in E(Y),\ \cu\in G_y.
\end{align*}
As $B_\cG$ is length-reducing, it is terminating. Furthermore, $F(\cG) = \rquot{\BSalp^*}{B_\cG}$.
A word $w\in \Del^*$ is called \emph{Britton-reduced} if no rule from $\cB_\cG$ can be applied to it. As $\cB_\cG$ is terminating, there is a Britton-reduced $\hat{w}$ with $\hat{w} \eqF w$ for every $w$. However, this $\hat{w}$ might not be unique as $\cB_\cG$ is not confluent in general. Still, the following crucial facts hold: 
\begin{lemma}[Britton's Lemma, \cite{britton63}]\label{lem:britton}
	Let $w \in \Del^*$ be Britton-reduced. 
	If $w \in_{F(\cG)} G_\PP$, then $w$ is the empty word or consists of a single letter of $G_\PP$.
	Moreover, if $w \eqF 1$, then $w=1$ (\ie $w$ is the empty word).	
\end{lemma}

\begin{lemma}\label{lem:britton2}
	If $v = h_0x_1\cdots g_{n-1}x_n h_{n}, w=g_0y_1\cdots g_{n-1}y_n g_{n} \in \Pi(\cG)$ with $v\eqF w$ are Britton-reduced, then $x_i=y_i$ for all $i$ and there are $\cu_i \in G_{y_i}$ for $1 \leq i \leq n$ such that
	\begin{align*}
		h_0 &=_{G_{\!\sor\:\!\!( \:\!\!y_{\;\!\!1}\!)}} g_0\:\! {(\cu_1^{-1})}^{\mathrlap{y_1}},\\ h_i &=_{G_{\!\tar\:\!\!(\:\!\!y_{\:\!\!i}\!)}} \cu_{i}^{\ov y_i} \:\! g_i \:\! {(\oi \cu_{i+1})}^{\mathrlap{y_{i+1}}}, \qquad\text{ for } 1\leq i \leq n-1, \text{ and } \\ h_n &=_{G_{\!\tar\:\!\!( \:\!\!y_{\:\!\!n}\!)}} \cu_n^{\ov y_n} \:\! g_n.
	\end{align*}
\end{lemma}

Using \prettyref{lem:britton} one obtains a decision procedure for the word problem if the subgroup membership problem of $G_y^y$ in $G_{\sor(y)}$ is decidable, the word problem of $G_\PP$ is decidable for some $\PP\in V(Y)$, and the isomorphisms $G_y^y \to G_y^{\ov y}$ are effectively computable for all $y\in E(Y)$.
However, this does not imply any bound on the complexity. The problem is that~-- even if all computations can be performed efficiently~-- the blow up due to the calculations of the isomorphisms $G_y^y \to G_y^{\ov y}$ might prevent an efficient solution of the word problem in the fundamental group. An example is the Baumslag group $\BG = \genr{a,t,b}{tat^{-1}= a^2, bab^{-1}=t}$, which is an HNN extension of the Baumslag-Solitar group $\BS12$. For \BG, the straightforward algorithm of applying Britton reductions, leads to a non-elementary running time. However, in \cite{muw11bg} it is shown that the word problem still can be solved in polynomial time.

For Baumslag-Solitar groups, the straighforward application of Britton reductions yields a polynomial time algorithm if the exponents are stored as binary integers.

\paragraph{Generalized Baumslag-Solitar Groups.}
A \emph{generalized Baumslag-Solitar group} (\emph{GBS group}) is a fundamental group of a finite graph of groups with only infinite cyclic vertex and edge groups. 
That means a GBS group is completely given by a finite graph $Y$ and numbers $\alp_\y ,\bet_\y \in \Z\setminus\oneset{0}$ for $\y \in E(Y)$ such that $\alp_\y=\bet_{\ov \y}$. For $a \in V(Y)$ we write $G_a=\gen{a}$. Then we have
$$ F(\cG) = \genr{V(Y),E(Y)}{\ov \y \y = 1,\, \y b^{\bet_\y}\ov\y = a^{\alp_\y} \text{ for } \y \in E(Y),\, a=\sor(\y),\, b = \tar(\y)} $$
and $G = \pi_1(\cG,a)\leq F(\cG)$ for any $a \in V(Y)$ as in \prettyref{def:fgogP}. (Note that $V(Y) \cup E(Y)$ generates $F(\cG)$ as a group, but in general, \emph{not} as a monoid.)

As we have seen in \prref{ex:bs}, Baumslag-Solitar groups \BS{p}q are the special case that $Y$ consists of one vertex and one loop $\y$ with $\alp_\y = p$, $\bet_\y=q$. 

\section{The Word Problem}\label{sec:wp}
In \cite{Robinson93phd}, Robinson showed that the word problem of non-cyclic free groups is $\NC^1$-hard. Hence, for non-solvable GBS groups, we cannot expect the word or conjugacy problem to be in $\TC$ since they contain a free group of rank two. For ordinary Baumslag-Solitar groups, the word problem has recently been shown to be in $\NC^2$ \cite{Kausch14bs}. In the author's dissertation \cite{Weiss15diss}, this is improved to \LOGDCFL~-- which means that it is \L-reducible to a deterministic context-free language. Here we aim for a \L algorithm~-- or, more precisely, for a \TC many-one reduction to the word problem of the free group $F_2$.

Let $G = \pi_1(\cG, a)$ be a fixed GBS group given by a graph $Y$ and numbers $\alp_\y ,\bet_\y \in \Z\setminus\oneset{0}$ for $\y \in E(Y)$ and $a\in V(Y)$. Our alphabet is $\Del = E(Y) \cup \set{a^k}{a \in V(Y),\, k \in \Z}$~-- for simplicity we allow $k=0$ and identify the letter $a^0$ with the empty word. We say that a word or $\cG$-factorization $w$ is represented in \emph{binary} if the numbers $k$ are written as binary integers (using a variable number of bits)~-- in the following we always assume this binary representation.

It turns out to be more convenient to work outside of $G$ and to consider arbitrary $\cG$-factorizations $w \in \Pi(\cG)$. Recall that a $\cG$-factorization of some group element is a word $$w = a_0^{k_0}\y_1a_1^{k_1}\cdots \y_na_n^{k_n}$$ with $a_i = \tar(\y_i) =  \sor(\y_{i+1})$ for $0 < i < n$, $a_n = \tar(\y_n) = a_0 = \sor(\y_{1})$, and $k_i \in \Z$. In the following, we always write $a_i$ as shorthand of $ \sor(\y_{i+1})$. 

\begin{lemma}\label{lem:horocycl}
	Let $w = a_0^{k_0}\y_1a_1^{k_1}\cdots \y_na_n^{k_n} \in \Pi(\cG)$. 
	If $w \in_{F(\cG)} \gen{a_0}$, then we have $w\eqF a_0^k$ for 
	$$k= \sum_{\nu = 0}^n k_\nu \cdot \prod_{\mu = 1}^\nu \frac{\alp_\mu}{\beta_\mu} $$
	where $\alp_\mu = \alp_{\y_\mu}$ and $\bet_\mu = \bet_{\y_\mu}$ for $1\leq \mu \leq n$.
\end{lemma} 
\begin{proof}
	If $w=a_0^{k_0}$, then the formula is obviously correct. Hence, let $n>0$. Then by \prettyref{lem:britton}, all the edges $\y_i$ can be cancelled by Britton reductions. In particular, we can find some $1<i\leq n$ such that $w= a_0^{k_0}\,\y_1 w' \y_i \,w''$ with $\y_1 = \ov \y_i$ and $w' = a_1^{k_1}\y_2\cdots a_{i-1}^{k_{i-1}} \in_{F(\cG)} \gen{a_1^{\bet_1}}$ and $w'' = a_i^{k_i}\y_{i+1} \cdots a_n^{k_n} \in_{F(\cG)} \gen{a_i} = \gen{a_0}$.
		
	By induction, we have $w'\eqF a_1^{k'}$ and $w''\eqF a_0^{k''}$ where 
	\begin{align*}
		k'&= \sum_{\nu = 1}^{i-1} k_\nu \cdot\prod_{\mu = 2}^\nu \frac{\alp_\mu}{\beta_\mu},
		& k''&= \sum_{\nu = i}^n k_\nu \cdot\!\! \prod_{\mu = i+1}^\nu \frac{\alp_\mu}{\beta_\mu}.
	\end{align*}
	Since $\y_1 w' \y_i \in_{F(\cG)} \gen{a_0}$, we have $\y_1 w' \y_i \eqF a_0^{\frac{\alp_1}{\bet_1}\cdot k'}$ and $\prod_{\mu = 1}^i \frac{\alp_\mu}{\beta_\mu} =1$. Hence,
	\begin{align*}
		k&= k_0 + \frac{\alp_1}{\beta_1} k' + k''
		=\sum_{\nu = 0}^n k_\nu \cdot\prod_{\mu = 1}^\nu \frac{\alp_\mu}{\beta_\mu}.
	\end{align*}
\end{proof}

For the rest of this section, we let $w = a_0^{k_0}\y_1a_1^{k_1}\cdots \y_na_n^{k_n}$ be a $\cG$-factorization given in binary. For $0\leq i \leq j \leq n$, we define 
\begin{align}
	\nonumber w_{i,j} &= a_i^{k_i}\y_{i+1}a_{i+1}^{k_{i+1}}\cdots \y_ja_j^{k_j}\\
	\qquad k_{i,j} &= \sum_{\nu = i }^j k_\nu \cdot\!\!\! \prod_{\mu = i + 1}^\nu\frac{\alp_\mu}{\beta_\mu} \qquad\in \Q\label{eq:kij}
\end{align}
analogously to $k$ in \prettyref{lem:horocycl} where again $\alp_\mu = \alp_{\y_\mu}$ and $\bet_\mu = \bet_{\y_\mu}$ for $1\leq \mu \leq n$. Note that we do not assume that $a_i^{k_i} \y_{i+1} a_{i+1}^{k_{i+1}} \cdots \y_j a_j^{k_j}$ lies in $\gen{a_i}$ -- yet the numbers $k_{i,j}$ will play an important role in what follows. In particular, with the notation of \prettyref{lem:horocycl}, we have $k = k_{0,n}$.
Moreover, by \prettyref{lem:horocycl}, we have

\begin{lemma}\label{lem:horocycl2} 	
	\hfil$w_{i,j} \in_{F(\cG)} \gen{a_i}$ if and only if $w_{i,j} \eqF a_i^{k_{i,j}}$.\hfill
\end{lemma} 

\begin{lemma}\label{lem:kijtc}
	The numbers $k_{i,j}$ (as fractions of binary integers) can be computed by a uniform family of \Tc0 circuits~-- even if the numbers $\alp_\y $, $\bet_\y $ are part of input.
\end{lemma}
\begin{proof}
	\proc{Iterated Addition} and \proc{Iterated Multiplication} are in \TC, see \prettyref{thm:divisionTC}; hence, the rational numbers $k_{i,j}$ can be computed in \TC according to \prref{eq:kij}. 
	Be aware that we do not require that the fractions are reduced.
\end{proof}

Now, pick some orientation $D\sse E(Y)$ of the edges (for every pair $\y$, $\ov \y$ choose exactly one of them to be in $D$). Consider the canonical map $\rho: G \to \Z^{D}$ onto the abelianization of the subgroup generated by the edges, which is defined by $a\mapsto 0$ for $a \in V(Y)$ and $\y\mapsto e_\y $, $\ov \y\mapsto -e_\y $ for $\y \in D$ (where $e_\y $ is the unit vector having $1$ at position $\y$ and $0$ otherwise). With other words $\rho$ counts the exponents of the edges. 
Consider the following observations:
\begin{itemize}
	\item If $w\eqF 1$, then every edge $\y$ in $w$ can be canceled with some $\ov \y$ by Britton reductions.
	\item Consider a factor $\y v\ov \y$ for some word $v$. If $\y$ cancels with $\ov \y$, then necessarily we have $\rho(v) = 0$, \ie all edges occurring in between have exponent sum zero.
\end{itemize}

Now, the idea is to introduce colors and assign them to the letters $\y_i $ such that $\y_i$ and $\y_j$ get the same color only if they potentially might cancel.		
In order to do so, we start by defining a relation $ {\sim_{\scriptscriptstyle\!\cC}} \subseteq\oneset{1,\dots, n} \times \oneset{1,\dots, n}$ and set
\begin{align*}
	i \sim_{\scriptscriptstyle\!\cC} j \text{ if and only if } \y_i = \ov \y_j \text{ and } &
	\begin{cases}
		\rho(w_{i,j-1}) = 0\text { and } k_{i,j-1} \in \bet_i\Z&\text{if } i < j,\\
		\rho(w_{j,i-1}) = 0 \text { and } k_{j,i-1} \in \bet_j\Z&\text{if } j < i.
	\end{cases} 
\end{align*}
Thus, $\sim_{\scriptscriptstyle\!\cC}$ is symmetric and we have $i \not\sim_{\scriptscriptstyle\!\cC}i$ for all $i$.
Informally speaking, we have $i\sim_{\scriptscriptstyle\!\cC} j$ if and only if everything in between vanishes in the abelian quotient $ \Z^{D}$ and $\y_i $ and $ \y_j$ cancel given that everything in between cancels to something in $\gen{a_i}$ (the latter is a consequence of \prref{lem:horocycl2}).

\begin{lemma}\label{lem:colorwelldef}		
	If
	$i\sim_{\scriptscriptstyle\!\cC} \ell$, $\ell \sim_{\scriptscriptstyle\!\cC} m$, and $m \sim_{\scriptscriptstyle\!\cC} j$, then also $ i\sim_{\scriptscriptstyle\!\cC} j$. 
\end{lemma}

\begin{proof}
\newcommand{\aA}{{\lambda_1}}
\newcommand{\bb}{{\lambda_2}}
\newcommand{\cc}{{\lambda_3}}
\newcommand{\dd}{{\lambda_4}}
If two of the indices $i,j,\ell,m$ coincide, what can be the case only if $i= m$ or $\ell = j$, we are done. 
Otherwise, we have to show that $\y_i = \ov \y_j$, $\rho(w_{i,j-1})=0$, and $k_{i,j-1} \in \beta_i\Z$ (resp. $\rho(w_{j,i-1})=0$ and $k_{j,i-1} \in \beta_j\Z$ for $j<i$).
We have $\y_i = \ov \y_\ell = \y_m = \ov \y_j$.

In order to see the other two conditions, we put the indices $i,j,\ell,m$ in ascending order. That means we fix $\aA < \bb < \cc < \dd$ such that $\oneset{\aA,\bb,\cc,\dd} = \oneset{ i,j,\ell,m}$. There are three situations to consider, as depicted in \prettyref{fig:threebracket}:
\begin{enumerate}
	\item $\y_\aA = \y_\bb$ and $\y_\cc = \y_\dd = \ov \y_\aA$, \label{uudd}
	\item $\y_\aA = \y_\cc$ and $\y_\bb = \y_\dd = \ov \y_\aA$, \label{udud}
	\item $\y_\aA = \y_\dd$ and $\y_\bb = \y_\cc = \ov \y_\aA$.\label{uddu}
\end{enumerate}
\begin{figure}[bht]
	\begin{center}
		\begin{tikzpicture}[scale=.6]
		\newcommand{\labeloffset}{-1.2}
		\draw[line width=.8pt,cap=round,shorten >=2pt,shorten <=2pt]
		(0,0) -- ++(0.6,1);	
		
		\draw[line width=.8pt,cap=round,shorten >=2pt,shorten <=2pt]
		(2.4,0) -- ++(0.6,1);
		
		\draw[line width=.8pt,cap=round,shorten >=2pt,shorten <=2pt]
		(5,0) -- ++(-0.6,1);
		
		\draw[line width=.8pt,cap=round,shorten >=2pt,shorten <=2pt]
		(7.4,0) -- ++(-0.6,1);
		
		\draw (.3,\labeloffset) node[above] {$\aA$};
		\draw (2.7,\labeloffset) node[above] {$\bb$};
		\draw (4.7,\labeloffset) node[above] {$\cc$};
		\draw (7.1,\labeloffset) node[above] {$\dd$};
		
		\draw [] (0.65,1.05) to [ncbar=3mm] (4.45,1.05); 
		\draw [] (3.1,1.05) to [ncbar=1.5mm] (4.3,1.05); 
		\draw [] (2.95,1.05) to [ncbar=4.5mm] (6.75,1.05); 
		\draw [] (0.5,1.05) to [ncbar=6mm] (6.9,1.05); 
		\end{tikzpicture}
	\end{center}
	\begin{center}
		\begin{tikzpicture}[scale=.6]
		\newcommand{\labeloffset}{-1.2}
		\draw[line width=.8pt,cap=round,shorten >=2pt,shorten <=2pt]
		(0,0) -- ++(0.6,1);	
		
		\draw[line width=.8pt,cap=round,shorten >=2pt,shorten <=2pt]
		(3,0) -- ++(-0.6,1);
		
		\draw[line width=.8pt,cap=round,shorten >=2pt,shorten <=2pt]
		(4.4,0) -- ++(0.6,1);	
		
		\draw[line width=.8pt,cap=round,shorten >=2pt,shorten <=2pt]
		(7.4,0) -- ++(-0.6,1);
		
		\draw (.3,\labeloffset) node[above] {$\aA$};
		\draw (2.7,\labeloffset) node[above] {$\bb$};
		\draw (4.7,\labeloffset) node[above] {$\cc$};
		\draw (7.1,\labeloffset) node[above] {$\dd$};
		
		\draw [] (0.7,1.05) to [ncbar=1.5mm] (2.4,1.05); 
		\draw [] (3,0) to [ncbar=-1.5mm] (4.4,0); 
		\draw [] (5,1.05) to [ncbar=1.5mm] (6.7,1.05); 
		\draw [] (0.55,1.05) to [ncbar=3.5mm] (6.85,1.05); 
		\end{tikzpicture}
	\end{center}
	
	\begin{center}
		\begin{tikzpicture}[scale=.6]
		\newcommand{\labeloffset}{-1.2}
		\draw[line width=.8pt,cap=round,shorten >=2pt,shorten <=2pt]
		(0,0) -- ++(0.6,1);	
		
		\draw[line width=.8pt,cap=round,shorten >=2pt,shorten <=2pt]
		(3,0) -- ++(-0.6,1);	
		
		\draw[line width=.8pt,cap=round,shorten >=2pt,shorten <=2pt]
		(5,0) -- ++(-0.6,1);
		
		\draw[line width=.8pt,cap=round,shorten >=2pt,shorten <=2pt]
		(7,0) -- ++(0.6,1);
		
		\draw (.3,\labeloffset) node[above] {$\aA$};
		\draw (2.7,\labeloffset) node[above] {$\bb$};
		\draw (4.7,\labeloffset) node[above] {$\cc$};
		\draw (7.1,\labeloffset) node[above] {$\dd$};
		
		\draw [] (0.7,1.05) to [ncbar=1.5mm] (2.4,1.05); 
		\draw [] (5,0) to [ncbar=-1.5mm] (6.95,0); 
		\draw [] (3,0) to [ncbar=-3.5mm] (7.1,0); 
		\draw [] (0.55,1.05) to [ncbar=3.5mm] (4.45,1.05); 
		\end{tikzpicture}
	\end{center}
	\caption{Three different situations. The four pairings of each situation are depicted as brackets.}\label{fig:threebracket}
\end{figure}
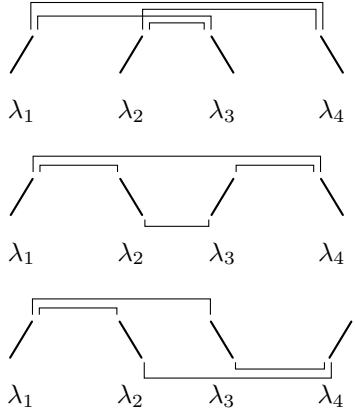
	
All these cases have in common that there are exactly four pairings $\oneset{\lambda_r, \lambda_s}$ with $\y_{\lambda_r} = \ov \y_{\lambda_s}$, and these four pairings correspond to the four pairings $\oneset{i, \ell}$, $\oneset{\ell, m}$, $\oneset{m, j}$,  and $\oneset{i, j}$. In each case, the conditions $\rho(w_{\lambda_r, \lambda_s-1})=0$ and $k_{\lambda_r, \lambda_s-1} \in \beta_{\lambda_r}\Z$ hold for three of the $\oneset{\lambda_r, \lambda_s}$, and we have to show it for the fourth.
				
In case \ref{uudd}, we have			
\begin{align*}
	\rho(w_{\aA,\dd-1}) = \rho(w_{\aA,\cc-1}) + \rho(w_{\bb,\dd-1}) - \rho(w_{\bb,\cc-1}).
\end{align*}
Thus, since three of these vectors are zero, so is the fourth (\ie we have shown that $\rho(w_{i,j-1})=0$ resp. $\rho(w_{j,i-1})=0$). In particular, we have $\rho(w_{\aA,\bb}) =  \rho(w_{\aA,\dd-1}) - \rho(w_{\bb,\dd-1}) = 0$. Hence,  $$\prod_{\mu = \aA + 1}^\bb\frac{\alp_\mu}{\bet_\mu} = \prod_{\y \in D} \left(\frac{\alp_\y}{\bet_\y}\right)^{\rho(w_{\aA,\bb})_\y} = 1,$$ where $\rho(w_{\aA,\bb})_\y$ denotes the component of the vector belonging to $\y$ (recall $D\sse E(Y)$ is the orientation)~-- the first equality is because $\frac{\alp_{\ov\y}}{\bet_{\ov\y}} = \left(\frac{\alp_\y}{\bet_\y}\right)^{-1}$ and $\rho(w_{\aA,\bb})_\y$ simply counts the number of occurrences of $\y$ (positive) and $\ov y$ (negative) in $w_{\aA,\bb}$. It follows that
\begin{align*}
	k_{\aA,\dd-1} &=\sum_{\nu = \aA }^{\dd -1} k_\nu \cdot\!\!\! \prod_{\mu = \aA + 1}^\nu\frac{\alp_\mu}{\bet_\mu}\\
	&= \sum_{\nu = \aA }^{\cc -1} k_\nu \cdot\!\!\! \prod_{\mu = \aA + 1}^\nu \frac{\alp_\mu}{\bet_\mu} \;+\, 	
	\sum_{\nu = \bb }^{\dd -1} k_\nu \cdot\!\!\! \prod_{\mu = \aA + 1}^\nu \frac{\alp_\mu}{\bet_\mu} \;-\, 
	\sum_{\nu = \bb }^{\cc -1} k_\nu \cdot\!\!\! \prod_{\mu = \aA + 1}^\nu\frac{\alp_\mu}{\bet_\mu}\\
	&= \sum_{\nu = \aA }^{\cc -1} k_\nu \cdot\!\!\! \prod_{\mu = \aA + 1}^\nu \frac{\alp_\mu}{\bet_\mu} \;+\, 	
	\sum_{\nu = \bb }^{\dd -1} k_\nu \cdot\!\!\! \prod_{\mu = \bb + 1}^\nu \frac{\alp_\mu}{\bet_\mu} \;-\, 
	\sum_{\nu = \bb }^{\cc -1} k_\nu \cdot\!\!\! \prod_{\mu = \bb + 1}^\nu\frac{\alp_\mu}{\bet_\mu}\\
	&= k_{\aA, \cc-1 } + k_{\bb,\dd-1} - k_{\bb, \cc-1}.\vphantom{k^k}
\end{align*}
Hence, since three of them are in $\bet_\aA\Z = \bet_\bb \Z$, so is the fourth.

The other cases follow with the same arguments: in case \ref{udud} we have		
\begin{align*}
	\rho(w_{\aA,\dd-1}) &= \rho(w_{\aA,\bb-1}) + \rho(\y_\bb) + \rho(w_{\bb,\cc-1})+ \rho(\y_\cc) + \rho(w_{\cc,\dd-1})\\
	&= \rho(w_{\aA,\bb-1}) + \rho(w_{\bb,\cc-1})+ \rho(w_{\cc,\dd-1})
\end{align*}
because $\y_\bb = \ov \y_\cc$, what again implies that all of them are zero. Like in the first case, we have $\prod_{\mu = \aA + 1}^\bb\frac{\alp_\mu}{\bet_\mu} = \frac{\alp_\bb}{\bet_\bb}$ (because $\rho(w_{\aA,\bb-1}) = 0$) and $\prod_{\mu = \aA + 1}^\cc\frac{\alp_\mu}{\bet_\mu} =1$ (because $\rho(w_{\aA,\cc}) = \rho(w_{\aA, \dd-1}) - \rho(w_{\cc, \dd - 1}) = 0$). It follows that
\begin{align*}
	k_{\aA,\dd-1} &= \sum_{\nu = \aA }^{\dd -1} k_\nu \cdot\!\! \prod_{\mu = \aA + 1}^\nu\frac{\alp_\mu}{\bet_\mu} \\
	&= \sum_{\nu = \aA }^{\bb -1} k_\nu \cdot\!\!\! \prod_{\mu = \aA + 1}^\nu\frac{\alp_\mu}{\bet_\mu} \;+\, 
	\frac{\alp_\bb}{\bet_\bb} \cdot \sum_{\nu = \bb }^{\cc -1} k_\nu \cdot\!\!\! \prod_{\mu = \bb+1}^\nu \frac{\alp_\mu}{\bet_\mu} \;+\, 
	\sum_{\nu = \cc }^{\dd -1} k_\nu \cdot\!\!\! \prod_{\mu = \cc + 1}^\nu \frac{\alp_\mu}{\bet_\mu}\\
	& =k_{\aA, \bb-1 } + \frac{\alp_\bb}{\bet_\bb} \cdot k_{\bb,\cc-1} + k_{\cc, \dd-1}.
\end{align*}
Since $\y_\aA = \y_\cc = \ov \y_\bb$, we have $\alp_\bb = \bet_\aA$ and $\bet_\cc = \bet_\aA$. That means we have $\frac{\alp_\bb}{\bet_\bb} \cdot k_{\bb,\cc-1} \in \bet_\aA \Z$ if and only if $k_{\bb,\cc-1} \in \bet_\bb \Z$, and $k_{\cc, \dd-1}	\in \bet_\aA \Z$ if and only if $k_{\cc, \dd-1} \in \bet_\cc \Z$. Thus, since for three of the $k_{\lam,\lam'}$ we have $k_{\lam,\lam'} \in \bet_\lam \Z$, this is true also for the fourth. 

Finally, in case \ref{uddu}, because of $\y_\bb = \y_\cc$, we have			
\begin{align*}
	\rho(w_{\aA,\cc-1}) - \rho(w_{\aA,\bb-1}) &= \rho(w_{\bb - 1,\cc-1}) \\ 
	&= \rho(w_{\bb ,\cc}) \qquad \\
	&= \rho(w_{\bb,\dd-1}) - \rho(w_{\cc,\dd-1}).
\end{align*} Therefore, they are all $0$. As before, $\rho(w_{\aA,\bb-1}) = 0$ implies that $\prod_{\mu = \aA + 1}^{\bb - 1} \frac{\alp_\mu}{\bet_\mu} = 1$ and $ \rho(w_{\bb ,\cc}) = 0$ implies that $\prod_{\mu = \bb + 1}^{\cc} \frac{\alp_\mu}{\bet_\mu} = 1$. Thus, we have
\begin{align*}
	k_{\aA,\cc-1} - k_{\aA, \bb-1 } 
	&= \sum_{\nu = \aA }^{\cc -1} k_\nu \cdot\!\!\! \prod_{\mu = \aA + 1}^\nu\frac{\alp_\mu}{\bet_\mu} \;-\,  \sum_{\nu = \aA }^{\bb -1} k_\nu \cdot\!\!\! \prod_{\mu = \aA + 1}^\nu\frac{\alp_\mu}{\bet_\mu} \\
	&= \frac{\alp_\bb}{\bet_\bb} \cdot\sum_{\nu = \bb }^{\cc -1} k_\nu \cdot\!\!\! \prod_{\mu = \bb+1}^\nu \frac{\alp_\mu}{\bet_\mu} \\
	&= \frac{\alp_\bb}{\bet_\bb} \cdot \left(\sum_{\nu = \bb }^{\dd -1} k_\nu \cdot\!\!\! \prod_{\mu = \bb+1}^\nu \frac{\alp_\mu}{\bet_\mu} \;-\,  \sum_{\nu = \cc }^{\dd -1} k_\nu \cdot\!\!\! \prod_{\mu = \cc+1}^\nu \frac{\alp_\mu}{\bet_\mu} \right)\\
	&= \frac{\alp_\bb}{\bet_\bb}\cdot\left( k_{\bb,\dd-1} - k_{\cc, \dd-1}\right)
\end{align*}
with $\alp_\bb = \bet_\aA$ and $\bet_\cc = \bet_\bb$.
So, again since for three of the $k_{\lam,\lam'}$ we have $k_{\lam,\lam'} \in \bet_\lam \Z$, this is true also for the fourth. 
\end{proof} 
Now, we define a new relation
${\approx} \subseteq\oneset{1,\dots, n} \times \oneset{1,\dots, n}$ as $i\approx j$ if and only if there is some $\ell $ with $i\sim_{\scriptscriptstyle\!\cC} \ell$ and $\ell \sim_{\scriptscriptstyle\!\cC} j$. Moreover, we set $i \approx i$ for all $i$.

\begin{lemma}\label{lem:coloraq}
	$\approx$ is an equivalence relation.
\end{lemma}
\begin{proof}
By definition, $\approx$ is reflexive. Because $\sim_{\scriptscriptstyle\!\cC}$ is symmetric, $\approx$ is also symmetric. Transitivity follows from \prettyref{lem:colorwelldef}.
\end{proof}

Denote by $\Sig_w = \set{[i]}{i \in \oneset{1,\dots,n}}$ the set of equivalence classes of $\approx$. 
For $[i] \in \Sig_w$ define $\ov{[i]} = [j]$ if $i	\sim_{\scriptscriptstyle\!\cC} j$~-- if no such $j$ exists, we add a new element $\ov{[i]}$ to $\Sig_w$. From the definition of $\approx$ it follows that $\ov{\,\cdot\,}$ is well-defined. Moreover, we have $\ov{\ov{[i]}} = [i]$ and $\ov{[i]} \neq [i]$ for all $[i] \in \Sig_w$. In particular, $\Sig_w$ is an alphabet with fixed-point-free involution. We can think of each class $[i] \cup \ov{[i]}$ as a color assigned to the edges $y_i$. From the definition of $\sim_{\scriptscriptstyle\!\cC}$ and \prref{lem:horocycl2} it is clear that only edges with the same color can cancel.
Let $\Lambda_w\sse \Sig_w$ such that $\Sig_w = \Lambda_w \cup \ov \Lambda_w$ as a disjoint union, \ie for every pair $[i],\ov{[i]}$ exactly one of them is in $\Lambda_w$. Then we have $\Sig_w^*/\set{[i]\ov{[i]} = \ov{[i]}[i]=1 }{i \in \oneset{1,\dots,n}} = F_{\Lambda_w}$.
Now, we define
\newcommand{\Col}[1]{\cC(#1)}	
\begin{align*}
	\Col{w} &= [1]\cdots [n], & \Col{w_{i,j}} &= [i+1]\cdots [j].
\end{align*}

\begin{lemma}\label{lem:color}
	\hfil$w_{i,j} \in_{F(\cG)} \gen{a_i} \text{ if and only if } \Col{ w_{i,j}} =_{F(\Lambda_w)} 1$.$\qquad$ \hfill
\end{lemma}

Before we prove \prref{lem:color}, we present an example and some consequences.
\begin{example}
	Consider the group \BS{2}{3} and the word  
	\begin{align*}w &= \y a \y a \ov \y a^3 \y a \ov \y a \ov \y \,\y a^2 \ov \y.
		\intertext{Then we have}
		\Col{w} &= [1][2][3][4][5][6][7][8]\\
		&= [1][2][3]\ov{[3]}\ov{[2]}\ov{[1]}[1]\ov{[1]} =_{F_{\Lambda_w}}1.
	\end{align*}
	Indeed, consider for example the factor $\ov \y a^3 \y$. As $k_{3,3} = 3 \in 3\Z$, it follows that $3 \sim_{\scriptscriptstyle\!\cC} 4$ and thus $[4] = \ov{[3]}$; however, $2 \not\sim_{\scriptscriptstyle\!\cC} 3$ since $k_{2,2} = 1 \not\in 2\Z$, see \prref{fig:klammergebirge}.
	By \prref{lem:color}, we know that $w \in \gen{a}$.
\end{example}

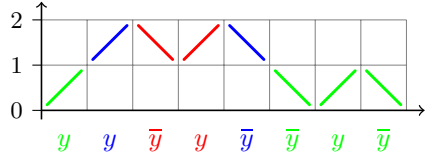
\begin{figure}[hbt]
	\begin{center}
		\begin{tikzpicture}[scale=.6]
		\draw[style=help lines] (0,2) grid (8,4);
		\draw[line width=0.6pt,->] (0,2) -- (8.4,2);
		\draw[line width=0.6pt,->] (0,1.8) -- (0,4.4); 
		\draw (0,2) -- (-0.2,2) node[left] {\small$0$};
		\draw (0,3) -- (-0.2,3) node[left] {\small$1$};
		\draw (0,4) -- (-0.2,4) node[left] {\small$2$};
		\draw[line width=1.1pt,cap=round,shorten >=3pt,shorten <=3pt, color=green]
		(0,2) -- (1,3);
		\draw (0.5,0.9) node[above] {\color{green}$ \y \vphantom{\ov \y}$};
		\draw[line width=1.1pt,cap=round,shorten >=3pt,shorten <=3pt, color=blue]
		(1,3) -- (2,4);
		\draw (1.5,0.9) node[above] {\color{blue}$\y\vphantom{\ov \y}$};
		\draw[line width=1.1pt,cap=round,shorten >=3pt,shorten <=3pt, color=red]
		(2,4) -- (3,3);
		\draw (2.5,0.9) node[above] {\color{red}$\ov \y$};
		\draw[line width=1.1pt,cap=round,shorten >=3pt,shorten <=3pt, color=red]
		(3,3) -- (4,4);
		\draw (3.5,0.9) node[above] {\color{red}$\y\vphantom{\ov \y}$};
		\draw[line width=1.1pt,cap=round,shorten >=3pt,shorten <=3pt, color=blue]
		(4,4) -- (5,3);
		\draw (4.5,0.9) node[above] {\color{blue}$\ov \y$};
		\draw[line width=1.1pt,cap=round,shorten >=3pt,shorten <=3pt, color=green]
		(5,3) -- (6,2);
		\draw (5.5,0.9) node[above] {$\color{green}\ov \y$};
		\draw[line width=1.1pt,cap=round,shorten >=3pt,shorten <=3pt, color=green]
		(6,2) -- (7,3);
		\draw (6.5,0.9) node[above] {$\color{green} \y\vphantom{\ov \y}$};	  	
		\draw[line width=1.1pt,cap=round,shorten >=3pt,shorten <=3pt, color=green]
		(7,3) -- (8,2);
		\draw (7.5,0.9) node[above] {\color{green}$\ov \y$};  	
		\end{tikzpicture}
	\end{center}
	\caption{$\rho(w)$ and $\Col{w}$ depicted graphically~-- each color represents one $[i] \cup \ov{[i]}$.}\label{fig:klammergebirge}
\end{figure}

As immediate consequences of Britton's Lemma, \prref{lem:horocycl}, and \prref{lem:color}, we obtain:

\begin{corollary}
	\hfil$w \eqF 1 \text{ if and only if } \Col{ w} =_{F(\Lambda_w)}1 \text{ and } k_{0,n}=0.$\hfill
\end{corollary}

\begin{corollary}\label{cor:Bred}
	For $w = a_0^{k_0}\y_1a_1^{k_1}\cdots \y_na_n^{k_n}$, let $[i_1]\cdots [i_j] \in \Sig_w^*$ be freely reduced with $\Col{w} = [1]\cdots [n] =_{F_{\Lambda_w}} [i_1]\cdots [i_j]$. Then the $\cG$-factorization
	$$\hat w = a_0^{k_{0,i_1-1}}\y_{i_1}a_{i_1}^{k_{i_1,i_2-1}} \cdots \y_{i_j}a_{i_j}^{k_{i_j,n}}$$
	is Britton-reduced and $w\eqF \hat w$.
\end{corollary}
Note that \prettyref{cor:Bred} is independent of the choice of the representatives of $[i_1],\dots, [i_j]$.

\begin{proof}[Proof of \prettyref{lem:color}]
	Let $w_{i,j} \in_{F(\cG)} \gen{a_i}$.
	By Britton's Lemma, we can write 
	$$w_{i,j}= a_i^{k_i}\, \y_{i+1} w_{i+1,\ell-1} \y_\ell \,w_{\ell,j} $$
	with $\y_\ell = \ov \y_{i+1}$, $w_{i+1,\ell-1} \in_{F(\cG)} \gen{a_{i+1}^{\bet_{i+1}}}$, and $w_{\ell,j} \in_{F(\cG)} \gen{a_i}$.
	By \prref{lem:horocycl2}, we have $k_{i+1,\ell -1} \in \bet_{i+1}\Z$. As also $\y_\ell = \ov \y_{i+1}$ and $\rho( w_{i+1,\ell-1})=0$, this implies $i +1 \sim_{\scriptscriptstyle\!\cC} \ell$. 			
	By induction, we know that $\Col{ w_{i+1,\ell-1}} =_{F(\Lambda_w)} \Col{ w_{\ell,j}} =_{F(\Lambda_w)} 1$.	
	Thus, we obtain $$ \Col{ w_{i,j}} = [i+1]\,\Col{ w_{i+1,\ell-1}} \, \ov{[i+1]}\,\Col{ w_{\ell,j}} =_{F(\Lambda_w)} 1.$$
	
	\smallskip	
	For the other direction let $ \Col{ w_{i,j}} =_{F(\Lambda_w)} 1$. Then $ \Col{ w_{i,j}}$ is not freely reduced and we can write it in the form $$\Col{ w_{i,j}} =[i +1]\,\Col{ w_{i+1,\ell-1}} \, [ \ell]\, \Col{ w_{\ell,j}}$$ for some $\ell$ with $[i +1] = \ov{[ \ell]}$ and $\Col{ w_{i+1,\ell-1}} =_{F(\Lambda_w)} \Col{w_{\ell,j}} =_{F(\Lambda_w)} 1$.
	
	By induction, we know that $w_{i+1,\ell-1} \in_{F(\cG)} \gen{a_{i+1}}$ and $w_{\ell,j} \in_{F(\cG)} \gen{a_\ell}$; thus, by \prref{lem:horocycl2}, $w_{i+1,\ell-1} \eqF a_{i+1}^{k_{i+1,\ell-1}}$ and $w_{\ell,j} \eqF a_\ell^{k_{\ell,j}}$.
	Since $[i +1]\sim_{\scriptscriptstyle\!\cC} [ \ell]$, we have $\y_{i+1} = \ov \y_\ell$ and $k_{i+1,\ell-1} \in \beta_{i+1}\Z$. As, in particular, $a_i = a_\ell$, we obtain			 
	\begin{align*}
		w_{i,j}&= a_i^{k_i}\;\y_{i+1} \, w_{i+1,\ell-1} \, \y_\ell \;w_{\ell,j}\\
		&\eqF a_i^{k_i}\, \y_{i+1}a_{i+1}^{k_{i+1,\ell-1}} \y_\ell\,a_\ell^{k_{\ell,j}} \\
		&= a_i^{k_i}\, a_{i}^{\frac{\alp_{i+1}}{\bet_{i+1}}k_{i+1,\ell-1}} \,a_i^{k_{\ell,j}} 
		\;\in \gen{a_i}.
	\end{align*}
\end{proof} 

Now, we are ready to describe a \TC-many-one reduction of the word problem for $\cG$-factorizations to the free group $F_2 = \gen{a,b}$. 
The input is a $\cG$-factorization $w$, the output some word in $\tilde w \in \oneset{a,\ov a, b, \ov b}^*$ such that $w \eqF 1$ if and only if $\tilde w =_{F_2} 1$. The circuit computes the following steps:
\begin{algorithm}\ \label{alg:WPTC}
	\begin{enumerate}		
		\item\label{WPik} Compute $k_{0,n}$. If $k_{0,n} \neq 0$, then output $a$ (or some arbitrary other non-identity element of $F_2$).
		\item\label{WPiiCol} Otherwise, compute and output an encoding of $\Col{ w}$ in $F_2$ as follows:
		\begin{enumerate}	
			\item\label{WPcolorsnake} For all pairs $i<j$ check independently in parallel whether $i\sim_{\scriptscriptstyle\!\cC} j$ in \TC:
			\begin{enumerate}
				\item check whether $\y_i = \ov \y_j$,
				\item compute $\rho(w_{i,j-1}) $ and check whether $\rho(w_{i,j-1})=0$,
				\item compute $k_{i,j-1}$, check whether $k_{i,j-1}\in \Z$ and, if yes, whether $\bet_i \mid k_{i,j-1}$.
			\end{enumerate}
			If all points hold, then $i\sim_{\scriptscriptstyle\!\cC} j$, otherwise not.
			
			\item For every index $i$ compute in parallel the smallest $j$ with $j \in [i] \cup \ov{[i]}$ as representative of $[i]$~-- depending on whether $j \in [i] $ or $j \in \ov{[i]}$ the corresponding output is $b^ja \ov b^{j}$ or $b^j\ov a \ov b^{j}$. 
			\item Concatenate all output words of the previous step.		
		\end{enumerate}
	\end{enumerate}		
\end{algorithm}
By \prref{lem:kijtc} and Hesse's result \prref{thm:additionTC}, step \ref{WPik} and \ref{WPiiCol} (a) 
can be computed in \TC. Steps \ref{WPiiCol} (b) and \ref{WPiiCol} (c) are straightforward in \TC. Indeed, the smallest $j \in [i] \cup \ov{[i]}$ satisfies the first order formula
$$\left(i = j \lor i \sim_{\scriptscriptstyle\!\cC} j \lor  \bigvee_k (i\sim_{\scriptscriptstyle\!\cC} k \land k\sim_{\scriptscriptstyle\!\cC} j)\right) \land \bigwedge_{k < j} \lnot \left(i\sim_{\scriptscriptstyle\!\cC} k \lor  \bigvee_\ell (i\sim_{\scriptscriptstyle\!\cC} \ell \land \ell \sim_{\scriptscriptstyle\!\cC} k)\right), $$
which describes an \AC circuit in the obvious way (see \cite{BarringtonIS90} for the general correspondence between circuits and formulas). 
Step \ref{WPiiCol} (c) can be seen as the application of a homomorphism of free monoids, what can be done in \TC (see \cite{LangeM98}). Thus, we have established a \TC many-one reduction to the word problem of $F_2$. 

Note that in none of the above steps the actual graph played a role~-- only the numbers $\alp_\y , \bet_\y$ were used. This is because, up to now, we assumed that the input is already given as $\cG$-factorization. But also the transformation of elements of $\pi_1(\cG,T)$ into $\cG$-factorizations can be done in \TC as we see in the next theorem, which proves \prref{thm:wp}.

\begin{theorem}\label{thm:wpgbs}
	Let $G = \pi_1(\cG,a)\cong \pi_1(\cG,T)$ be a GBS group with graph $Y$ and $\Del = E(Y) \cup \set{a^k}{a \in V(Y),\, k \in \Z}$. There is a many-one reduction computed by a uniform family of \Tc0-circuits from each of the problems
	\begin{enumerate}
		\item given a word $w \in \Del^*$, decide whether $w$ is a $\cG$-factorization and, if so, decide whether $w \eqF 1$.\label{wpP}
		\item given a word $w \in \Del^*$, decide whether $w =_{\pi_1(\cG,T)} 1$,\label{wpT}
	\end{enumerate}	
	to the word problem of the free group $F_2$. In particular, the word problem of $G$ is in \L.
\end{theorem}
\begin{proof}
	In order to decide whether $w$ is a $\cG$-factorization, one simply needs to verify whether $w$ is of the form $a_0^{k_0}\y_1a_1^{k_1}\cdots \y_na_n^{k_n}$ and then check whether $a_0 = a_n$ and $a_{i-1} = \sor(\y_i)$ and $\tar(\y_i) = a_i$ for all $1 \leq i \leq n$. This can be done in $\AC$. Then it remains to apply \prref{alg:WPTC}~-- which we already have seen to be in \TC.
	
	For \ref{wpT}, one needs to compute the isomorphism $\pi_1(\cG,T) \to \pi_1(\cG,a)$. For $a,b \in V(Y)$ let $T[a,b]$ denote the unique path from $a$ to $b$ in the spanning tree $T$. We read $T[a,b]$ as a group element. To compute the isomorphism, every letter $\y \in E(Y)$ has to be replaced by the word $T[a,\sor(\y)]\;\! \y \;\!T[\tar(\y),a]$ and every letter $b^k $  with $b \in V(Y),\, k \in \Z$ by $T[a,b]\;\!b^k\;\! T[b,a]$. This means we apply a  homomorphism of free monoids, what can be done in \TC (see \cite{LangeM98}). Moreover, this replacement produces a $\cG$-factorization as output. 
\end{proof}

\subsection{Computing Britton-reduced words}\label{sec:bred}
Before we consider the problem of computing Britton-reduced words, we focus on the analog problem in free groups, the computation of freely reduced words. Since already in the solution of the word problem free groups of arbitrary rank were appearing, we consider the alphabet $\Lambda$ as part of the input and assume that it is properly encoded over the binary alphabet $\oneset{0,1}$. In particular, we assume that the involution $\Lambda \cup \ov \Lambda \to \Lambda \cup \ov \Lambda$ can be computed in $\AC$~-- \eg\ by a bit-flip.

\begin{proposition}\label{prop:Fbred}
	The following problem is \AC-reducible to the word problem of $F_2$: given a finite alphabet $\Lambda$ and a word $w \in (\Lambda \cup \ov \Lambda)^*$, compute a freely reduced word $\hat w \in (\Lambda \cup \ov \Lambda)^*$ with $\hat w =_{F(\Lambda)} w$.
\end{proposition}
\begin{proof} 
	We follow a similar approach as for the solution of the word problem of GBS groups. For $w= w_1\cdots w_n$ with $w_i \in \Lambda \cup \ov \Lambda$, we set $w_{i,j} = w_{i + 1} \cdots w_j$.
	We define an equivalence relation ${\approx_{\scriptscriptstyle\!\cF}} \sse \oneset{1, \dots, n} \times \oneset{1,\dots, n}$ by
	\begin{align*}
		i \approx_{\scriptscriptstyle\!\cF} j \text{ if and only if } w_i = w_j \text{ and } &
		\begin{cases}
			w_{i,j} =_{F(\Lambda)} 1&\text{if } i < j,\\
			w_{j,i} =_{F(\Lambda)} 1&\text{if } j < i.
		\end{cases} 
	\end{align*}
	By using the embedding of $F_\Lambda$ into $F_2$, it can be checked in $\AC(F_2)$ for all pairs $i,j$ whether $i\approx_{\scriptscriptstyle\!\cF}j$. Furthermore, let us define a partial map 
	\begin{align*}
		&&&&\ov{\,\cdot\,}:\rquot{\oneset{1,\dots, n}}{\approx_{\scriptscriptstyle\!\cF}} &\to \rquot{\oneset{1,\dots, n}}{\approx_{\scriptscriptstyle\!\cF}}\\
		&&&&[i] &\mapsto \ov{[i]} = [j] &&\!\!\!\!\!\!\!\!\!\!\!\!\!\!\!\!\!\!\!\!\text{if there is some $j$ with } w_i = \ov w_j \text{ and }\\ 
		&&&&&&& \!\!\!\!\!\!\!\!\!\!\!\!\!\!\!\!\!\!\!\!  w_{i,j-1} =_{F_\Lambda} 1\ (\text{resp.\ } w_{j,i-1} =_{F_\Lambda} 1).
	\end{align*} 
	To see that this map is well defined, we have to verify two points: 
	\begin{enumerate}
		\item that the map $i \mapsto \ov{[i]}$ is well-defined;
		\item that  $\ov{[i]} = \ov{[j]}$ if $i \approx_{\scriptscriptstyle\!\cF} j $.
	\end{enumerate}
	For the first point, consider $i < j < k$ with $w_i = \ov w_j$,  $w_{i,j-1} =_{F_\Lambda} 1$ and $w_i = \ov w_k$,   $w_{i,k-1} =_{F_\Lambda} 1$. Then we have $w_j = \ov w_i = w_k$ and $w_{j,k} =_{F_\Lambda} (w_{i,j-1}w_j)^{-1} w_{i,k-1} w_k =_{F_\Lambda} 1$~-- hence, $j \approx_{\scriptscriptstyle\!\cF} k $. Likewise all other orderings of $i,j,k$ can be dealt with; hence, the image $\ov{[i]}$ is uniquely defined for each $i$.
	
	For the second point, let $i \approx_{\scriptscriptstyle\!\cF} j$ and $k \in \ov{[j]}$ with $i<j<k$. Then we have $w_i = w_j = \ov w_k$ and $w_{i,k-1} = w_{i , j  } w_{j, k- 1} 
	=_{F_\Lambda} 1$~-- that means $k \in \ov{[i]}$ and, thus, $\ov{[i]} = \ov{[j]}$. Again all other orderings of $i,j,k$ follow the same way.
	
	Since $\ov{\ov{[i]}} = [i]$ for all $i$, we have a well-defined partial involution $ \ov{\,\cdot\,}$.
	In the following, if $\ov{[i]}$ is not defined, we consider it to be the empty set.
					
	When looking at the indices in $[i] \cup \ov{[i]}$ in ascending order, indices from $[i]$ and indices from $\ov{[i]}$ always alternate. This is because if $w_{i,j} = 1$, then there must be some $k \in \oneset{i + 1, \dots, j - 1}$ such that $w_k$ cancels with $w_j$ by free reductions. In particular, $w_k = \ov w_j$ and $w_{k,j-1} = 1$. Thus, $k = \ov{[j]}$.
		
	Therefore, we have	$\abs{\abs{[i]}- \abs{\ov{[i]}}} \leq 1$ for all $i$. Moreover, if two equivalence classes $[i]$ and $\ov{[i]}$ have the same number of members, then all corresponding letters $w_j$ for $j \in [i] \cup \ov{[i]}$ can be canceled by free reductions. 
	On the other hand, if $\abs{[i]}- \abs{\ov{[i]}} = 1$, then after any sequence of free reductions, there remains still one letter $w_j$ for some $j \in [i] \cup \ov{[i]}$ which cannot be canceled. This is because a letter $w_i$ can only cancel with a letter $w_j$ if $w_i = \ov w_j$ and $w_{i,j-1} =_{F_\Lambda} 1$ (resp.\ $w_{j, i-1} =_{F_\Lambda} 1$)~-- with other words, $w_i$ can only cancel with letters $w_j$ for $j \in \ov{[i]}$.
	
	Thus, for each $i$ with $\abs{[i]}- \abs{\ov{[i]}} = 1$, denote by $j_{[i]}$ the maximal index in $[i]$. 
	Now, the freely reduced word $\hat w$ consists of exactly those $w_j$ with $j=j_{[i]}$ for some $i$. All other letters are deleted. Apart from the computation of $\approx_{\scriptscriptstyle\!\cF}$ and $ \ov{\,\cdot\,}$, everything can be done in \TC (with the same arguments as steps \ref{WPiiCol} (b) and \ref{WPiiCol} (c) of \prref{alg:WPTC}); hence, the whole procedure is in $\AC(F_2)$.
\end{proof}

\begin{corollary}\label{cor:BredL}\label{cor:bsbred}
	The following problems are \AC-reducible to the word problem of the free group $F_2$:\begin{enumerate}
		\item given a word $w \in \Del^*$, decide whether $w$ is a $\cG$-factorization and, if so, compute a Britton-reduced $\cG$-factorization $\hat w$ with $\hat{w} \eqF w$.
		\item given a word $w \in \Del^*$, compute a Britton-reduced $\cG$-factorization $\hat w$ with $\hat{w} =_{\pi_1(\cG,T)} w$.
	\end{enumerate}  Moreover, the number of bits required for $\hat{w}$ is linear in the number of bits of $w$.
\end{corollary}
\begin{proof}
	As in the proof of \prref{thm:wpgbs} it can be checked in \AC whether $w$ is a $\cG$-factorization (resp.\ a $\cG$-factorization can be computed from $w$ via the isomorphism $\pi_1(\cG,T) \to \pi_1(\cG,a)$ in \TC). Thus, we can assume that $w$ is a $\cG$-factorization.

	We can compute $\Col{w}\in \Lambda_w^*$ by step \ref{WPiiCol} of \prref{alg:WPTC} in \TC (or more precisely, a proper encoding of $\Col{w}$ over an alphabet of fixed size). By \prref{prop:Fbred}, a Britton-reduced word $\widehat{\Col{w}} = [i_1]\cdots [i_j] \in \Lambda_w^*$ can be computed in $\AC(F_2)$.
	By \prref{cor:Bred}, the desired output is $\hat w = a_0^{k_{0,i_1-1}}\y_{i_1}a_{i_1}^{k_{i_1,i_2-1}} \cdots \y_{i_j}a_{i_j}^{k_{i_j,n}}$; as before, it can be computed from $[i_1]\cdots [i_j]$ in \TC. According to \prref{eq:kij}, the number of bits of $k_{i,j}$ is linear in $j-i + \max \set{\log \abs{k_\nu}}{\nu \in \oneset{i, \dots, j}}$. Thus, the number of bits of $\hat{w}$ is linear in the number of bits of $w$.
\end{proof}

\subsection{Uniform versions of the word problem}\label{sec:uniWP}
It is not obvious what the uniform version of the word problem of GBS groups is (\ie a version of the word problem where the group is part of the input). Indeed, there are different ways how to define a uniform version of the word problem~-- and they lead to slightly different complexity bounds. We consider a uniform version of \prref{thm:wpgbs} \ref{wpP} and a uniform version of \prref{thm:wpgbs} \ref{wpT}.

In the uniform versions we assume that the graph of groups is given in a proper encoding. For instance we assume that the encoding consists of the numbers $\abs{V(Y)}$ and $\abs{E(Y)}$ and a list of tuples $(\y,\sor(\y), \tar(\y), \alp_\y , \bet_\y , \ov \y)$ for the edges. Here $\y, \ov \y \in \oneset{0,\dots, \abs{E(Y)}-1}$ and $\sor(\y), \tar(\y) \in \oneset{0,\dots, \abs{V(Y)}-1}$ and all numbers (also the $\alp_\y , \bet_\y )$ are encoded as binary integers using the same number of bits for all $\y$.
The graph of groups also defines the alphabet $\Del = E(Y) \cup \set{a^k}{a \in V(Y),\, k \in \Z}$. Recall that the integer exponents $k$ are represented in binary using a variable number of bits.

We say an encoding is \emph{valid}, if all tuples are properly formed, for every edge $\y$, there is an inverse edge $\ov \y$ satisfying $\sor(\y)=\tar(\ov \y)$ and $\alp_\y =\bet_{\ov \y}$, and the graph is connected.
\begin{corollary}\label{cor:uniWP1}
	The following problem is \TC-many-one-reducible to the word problem of $F_2$.
	Input: a valid encoding of a graph of groups $\cG$ and a word $w \in \Del^*$.
	Decide whether $w$ is a $\cG$-factorization and, if so, decide whether $w \eqF 1$.
\end{corollary}
Note that we need the promise in \prref{cor:uniWP1} that the input is a valid encoding of a graph of groups. Indeed, it cannot be checked whether the graph is connected in \TC unless $\TC = \L$ (by \cite{CookM87}, already connectivity for forests is \L-complete with respect to $\NC^1$-reductions~-- the reduction is actually a $\AC$ reduction\footnote{For a problem $\cP$ in \L, the reduction computes the full configuration graph of the \L Turing machine for $\cP$. By using a time-stamp and an additional sink vertex, every vertex except the sink vertex and the accepting configuration can be made to have out-degree exactly one. Now, for every edge an inverse edge can be introduced without changing connectivity properties~-- the resulting graph has precisely two connected components and in each component one vertex is known.}, see also \cite{JennerLM97}).

On the other hand, by the seminal paper by Reingold \cite{Reingold08}, connectivity of undirected graphs can be checked in \L. Hence, as the other points can be easily verified in $\AC$, it can be checked in \L whether an encoding of a graph of groups is valid.

\begin{proof}[Proof of \prref{cor:uniWP1}]
	We need to verify two things, namely, that for some word $w \in \Del^*$ it can be checked in \TC whether it is a $\cG$-factorization and, second, that \prref{alg:WPTC} is still in \TC also if the graph is part of the input.
	
	For the first point, it only needs to be checked whether $w$ is of the form $a_0^{k_0}\y_1a_1^{k_1}\cdots \y_na_n^{k_n}$ and, if so, whether $a_0 = a_n$ and $\sor(\y_i) = a_i$ and $\tar(\y_{i}) = a_{i+1}$ for all $i$. This can be done in $\AC$.
	
	\prref{alg:WPTC} is almost independent of the graph. Indeed, there is only the lookup of the numbers $\alp_\y , \bet_\y $, the check whether $\y=\ov x$ for edges $x,\y$, and the choice of the orientation $D$ in order to compute the homomorphism $\rho$. The first two points are straightforward and also the last point is no difficulty: for a pair $\y,\ov{\y}$, simply choose the one with smaller index in the list coding the graph to be in $D$.
\end{proof}

The uniform version of \prref{thm:wpgbs} \ref{wpT} is not so immediate. The difficulty lies in the computation of the paths $T[a,b]$. This problem is complete for \L under $\NC^1$ reductions \cite{CookM87} (and indeed under $\AC$ reductions as remarked above). Thus, together with the computation of the isomorphism $\pi_1(\cG,T) \to \pi_1(\cG,a)$, the algorithm of \prref{thm:wpgbs} \ref{wpT} is no longer a \TC many-one reduction (or at least it is not known whether it is). Still, we can prove the following result, which  together with \prref{cor:uniWP1} yields the proof for the first part of \prref{thm:uni}.
\begin{corollary}\label{cor:uniWP2}
	The following problem is complete for \L under $\AC$ reductions:
	Given a (valid) encoding of graph of groups $\cG$ with a spanning tree $T$ (given as list of edges) and a word $w \in \Del^*$.
	Decide whether $w =_{\pi_1(\cG,T)} 1$.
\end{corollary}
Note that the question whether $w =_{\pi_1(\cG,T)} 1$ depends on the spanning tree $T$. For instance assume that $w=\y_1$ consists of a single edge. Then we have $w =_{\pi_1(\cG,T)} 1$ if and only if $\y_1$ is part of the spanning tree. 
Therefore, we require $T$ to be part of the input, although 
by \cite{NisanT95} (together with \cite{Reingold08}), for a given graph a spanning tree can be computed in \L.

\begin{proof}
	Since apart from the computation of the isomorphism $\pi_1(\cG,T) \to \pi_1(\cG,a)$ (which is in \L by \cite{CookM87}), we are in the same situation as in \prref{cor:uniWP1}, it remains to prove the hardness part. 
	
	We reduce the following special version of \emph{Undirected Forest Accessibility} (see \cite{CookM87}) to our problem. The problem receives an undirected forest (\ie an acyclic graph) $\Gamma$ with precisely two connected components and three vertices $s,t,u \in V(\GG)$ as input such that $t$ and $u$ are in two different connected components~-- we may assume that the graph is given as list of tuples $(\y,\sor(\y), \tar(\y), \ov \y)$ representing edges where each tuple has the same bit-length. The question is whether $s$ and $t$ are connected by a path.
	
	In order to obtain an instance for the uniform word problem of GBS groups $(\cG,T, w)$, we take the input forest $\Gamma$ and assign to every edge $\y$ the numbers $\alp_\y = \bet_\y = 1$. That means each tuple $(\y,\sor(\y), \tar(\y), \ov \y)$ has to be replaced by $(\y,\sor(\y), \tar(\y), 1,1,\ov \y)$. As we assumed all tuples to have the same bit-length, this can be done hard-wired in the circuit.
	Finally, we create a new edge $\y_{tu}$ connecting $t$ and $u$ with
	$\alp_{\y_{tu}} = \bet_{\y_{tu}} = 2$. The spanning tree $T$ consists of all edges. The input word is $w=s\oi t$.
	
	Now, $\pi_1(\cG,T)$ is isomorphic to the amalgamated product $\gen{t} *_{t^2=u^2} \gen{u}$ and we have either $s=_{\pi_1(\cG,T)} t$ or $s=_{\pi_1(\cG,T)} u$~-- depending on the connected component of $\Gamma$ in which $s$ lies.
	In particular, $s=_{\pi_1(\cG,T)} t$ if and only if $s$ and $t$ are connected by some path in $\Gamma$.
\end{proof}

\newcommand{\ku}{u}

\section{The Conjugacy Problem}\label{sec:cp}
Decidability of conjugacy in Baumslag-Solitar groups was established by Anshel and Stebe \cite{AnshelS74}.
In \cite{Anshel76conjugacy} this was generalized to the special case of GBS groups where the graph $Y$ consists of only one vertex (\ie an HNN extension with several stable letters).
Later in \cite{Horadam84}, Horadam showed that the conjugacy problem is decidable in GBS groups if there is some constant $c \in \Z$ with $\alp_\y = c$ for all $\y\in E(Y)$. In \cite{HoradamF94}, this was further generalized to some other class of GBS groups which contains the linear GBS groups (with the generalization that they also considered infinite graphs); in \cite{Lockhart92}, Lockhart gave a solution for all GBS groups. Finally, in \cite{Beeker11thesis}, Beeker independently gave a solution of the conjugacy problem in all GBS groups.

Before we start with the solution of the conjugacy problem in GBS groups, we recall some general facts about conjugacy in fundamental groups of graphs of groups.

\subsection{Conjugacy and Graphs of Groups}\label{sec:gogconj}
Let $\cG$ again be an arbitrary graph of groups with graph $Y$ and $\PP \in V(Y)$.

\begin{lemma}\label{lem:path}
	Let $g,h \in\pi_1(\cG,\PP) \leq F(\cG)$. If $g\sim_{F(\cG)} h$, then already $g\sim_{\pi_1\:\!\!(\cG,\PP)} h$.
\end{lemma}
\begin{proof}
	Let $\phi: F(\cG) \to \pi_1(\cG,T)$ be the projection and $\psi: \pi_1(\cG,T) \to \pi_1(\cG,\PP)$ be the canonical isomorphism. If $z \in F(\cG)$ is a conjugator, then $\psi(\phi(z)) \in \pi_1(\cG,\PP)$ is also a conjugator.
\end{proof}

By \prettyref{lem:path}, instead of testing conjugacy in the fundamental group $\pi_1(\cG,\PP)$, we can test it in the larger group $F(\cG)$. This simplifies the algorithms substantially because for $\cG$-factorizations in $F(\cG)$ there is good notion of cyclically Britton-reduced elements.

Let $w = g_0y_1g_1 \cdots y_ng_{n} \in \Pi(\cG)$. We say that $v$ is a \emph{cyclic permutation} of $w$ if there are $u,u'\in \BSalp^*$ such that $w=uu'$ and $v=u'u$.
A word $w\in \BSalp^*$ is called \emph{cyclically Britton-reduced} if every cyclic permutation of $w$ is Britton-reduced.
That means $w$ is cyclically Britton-reduced if and only if $ww$ is Britton-reduced or $w \in G_a$ for some $a \in V(Y)$.
The following lemma provides a tool to compute cyclically Britton-reduced $\cG$-factorizations.
\begin{lemma}\label{lem:cycred}
	Let $w = g_0y_1g_1 \cdots y_ng_{n}\in \Pi(\cG)$ with $n\geq 1$ be Britton-reduced. Then for $$ y_{\floor{n/2+1}}g_{\floor{n/2+1}} \cdots y_ng_{n}g_0 y_1g_1 \cdots y_{\floor{n/2}}g_{\floor{n/2}} 	\RAS{*}{B_\cG} \hat w,$$ if $\hat w$ is Britton-reduced, then $\hat w$ is cyclically Britton-reduced and $w \sim \hat w$.
\end{lemma}

\begin{proof}
	It is clear that $w \sim \hat{w}$.
	If $\hat w$ does not contain any $y \in E(Y)$, we are done. In the other case, we have to show that $\hat{w} \hat{w}$ is Britton-reduced. When computing $\hat w$, Britton reductions may only occur in the middle; thus, we know that $y_{\floor{n/2+1}}$ is still present in $\hat w$. If $\hat{w} \hat{w}$ is not Britton-reduced, then the occurrence of $y_{\floor{n/2+1}}$ in the second factor $\hat{w}$ must cancel with something in the first factor. This can be either $y_{\floor{n/2+1}}$ or $y_{\floor{n/2}}$ depending on whether $y_{\floor{n/2}}$ has been canceled when computing $\hat{w}$. However, the first case would mean that $y_{\floor{n/2+1}}$ is self-inverse; the second case is a contradiction to the assumption that $w$ was Britton-reduced.
\end{proof}

\newcommand{\CU}{\cC}

Let $\CU$ denote the union of all $G_y^y$.
The following result is due to Horadam \cite{Horadam81}; it is the main tool for deciding the conjugacy problem. For amalgamated products, it first appeared in \cite{mks66}; the special case for HNN extensions is known as \emph{Collins' Lemma} \cite{Collins69}~-- see also \cite[Thm.\ IV.2.5]{LS01}. 

\begin{theorem}[Conjugacy Criterion, \cite{Horadam81}]\label{thm:collins}
	Let $w \in \Pi(\cG)$ be cyclically Britton-reduced. Then one of the following cases holds:
	\begin{enumerate}
		\item
		There is some $\PP\in V(Y)$ with $w\in G_\PP$ ($w$ is called \emph{elliptic}).
		\begin{enumerate}
			\item\label{coli} If $w\sim_{F(\cG)} \cu$ for some $\cu \in \CU$, then there exists a
			sequence of elements $\cu= \cu_0, \cu_1,\dots, \cu_m \in \CU$ such
			that $\cu_m \sim_{G_\PP} w$ and for every $i$ there is some $b_i\in \BSalp$ with $\cu_i = b_i\cu_{i-1}\ov b_i$.
			
			\item\label{colii} 
			If $w$ is not conjugate to any $\cu \in \CU$ and $w\sim_{F_\cG} v$ for some cyclically Britton-reduced $v$, then $v \in G_\PP$ and $v \sim_{G_\PP} w$. 
		\end{enumerate}
		\item\label{coliii} We have $w \not\in G_\PP$ for any $\PP\in V(Y)$ ($w$ is called \emph{hyperbolic}), \ie $w$ has the form
		$
		w = y_1g_1 \cdots y_ng_{n}
		$
		with $n \geq 1$. If $w$ is conjugate to a cyclically Britton-reduced $\cG$-factorization $v = x_1h_1 \cdots x_mh_{m}$, then $m=n$ and there are $i \in \oneset{1,\dots, n}$ and $\cu \in G_{y_i}^{y_i}\sse \CU$ such that
		$$v\eqF \cu \,y_ig_i \cdots y_ng_n\,y_1g_1\cdots y_{i-1}g_{i-1}\,\oi \cu,$$
		\ie $w$ can be transformed into $v$ by a cyclic permutation followed by a conjugation with an element of $\CU$.
	\end{enumerate}
\end{theorem}

\subsection{Conjugacy in GBS groups}\label{sec:cpgbs}
The input for the conjugacy problem are two words $v,w \in \Del^*$. As we have seen in the proof of \prref{thm:wpgbs}, we may assume that $v$ and $w$ are either words representing group elements of the fundamental groups with respect to some spanning tree $\pi_1(\cG,T)$ or $\cG$-factorizations of elements of $\pi_1(\cG,a)$.
In view of \prref{thm:collins}, a first step towards the solution of the conjugacy problem is the computation of cyclically Britton-reduced $\cG$-factorizations. 
By \prref{cor:BredL} we can compute Britton-reduced $\cG$-factorizations in $\AC(F_2)$. Thus, by \prref{lem:cycred}, also cyclically Britton-reduced $\cG$-factorizations can be computed in $\AC(F_2)$. 

Before we start to examine conjugacy, we need a technical lemma:

\begin{lemma}\label{lem:congruencesTC}
	Let $\cP$ be some fixed finite set of prime numbers. The following problem is solvable in \TC: Given $c_i, d_i\in \Z$ (in binary) for $i=0,\dots,n$ such that $d_i$ has only prime factors in $\cP$. Decide whether the system of congruences
	\begin{align*}
		x &\equiv c_i \mod d_i &\text{for } i=0,\dots,n
	\end{align*}
	has a solution.
\end{lemma}

\begin{proof}
	Since $\cP$ is finite, the following can be done for all $p \in \cP$ in parallel. 
	Considering only powers of $p$, the system of congruences transforms into 
	\begin{align}
		x &\equiv c_i \mod p^{e_i} &\text{for } i=0,\dots,n \label{eq:congruences}
	\end{align}
	where $e_i$ is maximal such that $p^{e_i}$ divides $d_i$. Such $e_i$ can be determined in \TC by checking whether $p^e$ divides $d_i$ for all $0\leq e \leq \log \abs{ d_i}$ in parallel using \prettyref{thm:divisionTC} for \proc{Integer Division}. 
	If there is some $i\neq j$ with $e_i \leq e_j$ and $c_i \not\equiv c_j \mod p^{e_i}$, then \prettyref{eq:congruences} obviously does not have a solution. Again this can be checked in parallel for all pairs $i,j$. If there is no such pair $i\neq j$, \prettyref{eq:congruences} is equivalent to a single congruence $x \equiv c \mod p^{e}$ where $e = \max_{i\in\{0,\dots,n\}} e_i$ and $c=c_i$ for the respective $i$. 
	
	If \prettyref{eq:congruences} has a solution for all $p\in \cP$, then there is a solution for the original congruence by the Chinese Remainder Theorem.
\end{proof}

\begin{proposition}\label{prop:cycredconj}
	The following problem is in \TC: Given two cyclically Britton-reduced hyperbolic $\cG$-factorizations $v,w\in \Pi(\cG)$ in binary representation, decide whether $v\sim_{F(\cG)} w$.
\end{proposition}

\begin{proof}
	By assumption, we are in case \ref{coliii} of the Conjugacy Criterion, \prettyref{thm:collins}.
	Let
	\begin{align*}
		v &= \y_1 a_1^{k_1} \cdots \y_n a_n^{k_n}
	\end{align*}
	be a $\cG$-factorization. By \prettyref{thm:collins} \ref{coliii} we know that if $v$ and $w$ are conjugate, then the underlying path $\y_1\cdots \y_n$ of $v$ is a cyclic permutation of the underlying path of $w$. Since in \TC all these cyclic permutation can be checked in parallel, we may assume that $w$ is of the form
	\begin{align*}
		w &= \y_1 a_1^{\ell_1} \cdots \y_n a_n^{\ell_n}.
	\end{align*}
	and (also by \prettyref{thm:collins} \ref{coliii})
	$$v\sim_{F(\cG)} w \iff \exists\, x\in\Z \text{ such that } a^xva^{-x}\eqF w.$$
	By \prettyref{lem:britton2}, we have $a^xva^{-x}\eqF w$ if and only if there are $x_1, \dots, x_n\in \Z$ such that 
	\begin{align*}
		x - \alp_1 x_1 &= 0, \\
		\bet_i x_i + k_i - \alp_{i+1} x_{i+1} &= \ell_i &\text{for } i=1,\dots,n-1,\\ 
		\bet_n x_n + k_n - x &= \ell_n.
	\end{align*}
	Like in \cite{Horadam84}, these equations imply that it is decidable whether $v$ and $w$ are conjugate. As we aim for a good complexity bound, we have to take a closer look. By solving these equations for $x_{i+1}$, we obtain
	\begin{align*}
		&&x_1 &= \frac{x}{\alp_1}, \\
		&& x_{i+1} &= \frac{k_i - \ell_i+ \bet_i x_i}{\alp_{i+1}} &\text{for } i=1,\dots,n-1,\\ 
		&& x &= k_n - \ell_n + \bet_nx_n .
	\end{align*}
	By induction follows
	\begin{align}
		x_{i} &= \frac{1}{\alp_{i}}\left(x \cdot \prod_{\mu=1}^{i-1}\frac{\bet_\mu}{\alp_\mu} \,+\, \sum_{\nu =1}^{i-1}(k_\nu - \ell_\nu)\cdot\!\!\! \prod_{ \mu = \nu +1 }^{i-1}\frac{\bet_\mu}{\alp_\mu}\right) &\text{for } i=1,\dots,n,\label{eq:yeq}\\ 
		\intertext{and the last equation becomes}
		x &= k_n - \ell_n + x \cdot \prod_{\mu=1}^{n}\frac{\bet_\mu}{\alp_\mu} \,+\, \sum_{\nu =1}^{n-1} (k_\nu - \ell_\nu)\cdot\!\!\! \prod_{\mu = \nu +1 }^{n}\frac{\bet_\mu}{\alp_\mu}. \label{eq:xeq}
	\end{align}
	We distinguish two cases: 
	
	First, assume that \prettyref{eq:xeq} has a unique solution.
	Then, the rational values $x_i$ are also determined uniquely and we have $v\sim_{F(\cG)} w$ if and only if $x$ and the $x_i$ are all integers.
	In this case, we have $\prod_{\mu=1}^{n}\frac{\bet_\mu}{\alp_\mu} \neq 1$ and
	\begin{align*}
		x &= \frac{ \sum_{\nu =1}^{n}(k_\nu - \ell_\nu)\cdot \prod_{\mu = \nu +1 }^{n}\frac{\bet_\mu}{\alp_\mu}}{1-\prod_{\mu=1}^{n}\frac{\bet_\mu}{\alp_\mu}}. 
	\end{align*}
	All occurring numbers are rationals; hence, they can be represented as fractions of binary integers. Since \proc{Iterated Multiplication} is in \TC (\prettyref{thm:divisionTC}), the products can be computed. A common denominator for the sums can be computed by \proc{Iterated Multiplication}, again. Thus, calculating the sum is just \proc{Iterated Addition} (\prettyref{thm:additionTC}). Let $c,d,e,f \in \Z$ be such that $\frac{c}{d} = \sum_{\nu =1}^{n} (k_\nu - \ell_\nu) \cdot \prod_{\mu = \nu +1 }^{n}\frac{\bet_\mu}{\alp_\mu}$ and $\frac{e}{f} = 1-\prod_{\mu=1}^{n}\frac{\bet_\mu}{\alp_\mu}$. In case $x = \frac{cf}{de}$ is an integer, we can determine this by applying Hesse's circuit for \proc{Integer Division} (\prettyref{thm:divisionTC}) to $cf$ and $de$. If $x$ is not an integer, we can notice that by multiplying the result of the division with $de$; if the result is not $cf$, there is no $x$ with $a^xva^{-x}\eqF w$.
	
	If $x$ is an integer, the numbers $x_i$ can be computed in \TC with the same technique, and it can be checked whether $x_i \in \Z$ for all $i$. Thus, we are done with the case that \prettyref{eq:xeq} has a unique solution.
	
	In the second case, we have $\prod_{\mu=1}^{n}\frac{\bet_\mu}{\alp_\mu} = 1$. Then \prettyref{eq:xeq} is equivalent to 
	\begin{align*}
		k_n - \ell_n  \,+\, \sum_{\nu =1}^{n-1} (k_\nu - \ell_\nu) \cdot\!\!\! \prod_{\mu = \nu +1 }^{n}\frac{\bet_\mu}{\alp_\mu} &=0.
	\end{align*}
	Again, this equality can be checked in \TC as before. If the equality does not hold, then there is no $x$ with $a^xva^{-x}\eqF w$. Otherwise, by \prettyref{eq:yeq}, we have $a^xva^{-x}\eqF w$ for $x \in \Z$ if and only if 
	\begin{align*}
		x_i = \frac{1}{\alp_{i}}\left(x \cdot \prod_{\mu=1}^{i-1}\frac{\bet_\mu}{\alp_\mu} \,+\, \sum_{\nu =1}^{i-1} (k_\nu - \ell_\nu) \cdot\!\!\! \prod_{\mu = \nu +1 }^{i-1}\frac{\bet_\mu}{\alp_\mu} \right)&\in \Z &\text{for all } i\in \oneset{1,\dots,n}. 
	\end{align*}
	By solving for $x$, we obtain
	\begin{align}
		x &\in \Z \:\cap\: \bigcap_{i=1}^n \left( \alp_{i} \cdot \prod_{\mu=1}^{i-1}\frac{\alp_\mu}{\bet_\mu}\right)\cdot\left( -\frac{1}{\alp_{i}}\sum_{\nu =1}^{i-1} (k_\nu - \ell_\nu) \cdot\!\!\!\prod_{\mu =\nu +1 }^{i-1}\frac{\bet_\mu}{\alp_\mu} + \Z \right). \label{eq:xconjBG}
	\end{align}
	Let $M\in \Z$ be the product of the denominators of all terms in this intersection and $c_i,d_i\in \Z$ such that 
	\begin{align*}
		\frac{c_i}{M} &= -\left(\prod_{\mu=1}^{i-1}\frac{\alp_\mu}{\bet_\mu}\right) \cdot \sum_{\nu =1}^{i-1} (k_\nu - \ell_\nu) \cdot\!\!\! \prod_{\mu = \nu +1 }^{i-1}\frac{\bet_\mu}{\alp_\mu} &&\text{and} \\
		\frac{d_i}{M} &= \alp_{i} \cdot \prod_{\mu=1}^{i-1}\frac{\alp_\mu}{\bet_\mu} &&\text{for } i=1,\dots,n. 
	\end{align*}
	In addition, we set $c_0 = 0$ and $d_0 = M$. 
	Now, \prettyref{eq:xconjBG} is equivalent to
	\begin{align}
		x &\in \frac{1}{M} \left(c_i + d_i \Z \right) &\text{for } i =  0,\dots, n. \nonumber
		\intertext{We substitute $Mx$ by $z$. Because of the choice of $c_0$ and $d_0$, the existence of an integer solution $x$ is equivalent to the system of congruences}
		z &\equiv c_i \mod d_i  &\text{for } i=0,\dots,n \label{eq:yconjBG}
	\end{align}
	having a solution. Let $\cP$ be the finite set of prime divisors of the $\alp_\y $ and $\bet_\y $ for $\y \in E(Y)$. As $M$ as well as the $d_i$s are products of the $\alp_\y $ and $\bet_\y $, they have only prime factors in $\cP$. Furthermore, as before, the numbers $c_i, d_i$ and $M$ can be computed in \TC. By \prettyref{lem:congruencesTC}, it can be checked in \TC whether \prettyref{eq:yconjBG} has a solution.
\end{proof}

Before we examine the conjugacy problem for elliptic elements in GBS group, we consider the special case of Baumslag-Solitar groups $\BS{p}q$, where the solution is straightforward.
\begin{proposition}\label{prop:cycredconjvertex}
	The following problem is in \TC: Given $v=a^k$ and $w=a^\ell$ with $k,\ell \in \Z$ given in binary, decide whether $v\sim_{\BS{p}q} w$.
\end{proposition}

\begin{proof}
	We are in case \ref{coli} or \ref{colii} of the Conjugacy Criterion, \prettyref{thm:collins}. Since a conjugation with $a$ has no effect and a conjugation with $\y^{\pm 1}$ multiplies the exponent by $\frac{p}{q}$ resp.\ $\frac{q}{p}$, we have
	\begin{align*}a^k \sim a^\ell \iff \exists\, j\in\Z \text{ such that } k\cdot\left(\frac{q}{p}\right)^j =\ell
		&\text{ and }\begin{cases} k \in p\Z,\, \ell \in q\Z, &\text{if } j > 0,\\
			k \in q\Z,\, \ell \in p\Z,&\text{if } j < 0.
		\end{cases} 
	\end{align*}
	Since we have $\abs{j} \leq \log_{\abs{q/p}}\max\oneset{\abs{k},\abs{\ell}}$ if such $j$ exists, only polynomially many (in the input size) values for $j$ need to be tested, what can be done in parallel. As \proc{Iterated Multiplication} and \proc{Integer Division} are in \TC (\cite{hesse01, HeAlBa02}, see \prettyref{thm:divisionTC}), we have concluded the proof of \prettyref{prop:cycredconjvertex}.
\end{proof}

Thus, for ordinary Baumslag-Solitar groups, we have solved the conjugacy problem completely by combining \prettyref{cor:bsbred} with \prettyref{prop:cycredconj} and \prettyref{prop:cycredconjvertex}. 
	
For arbitrary GBS groups, it remains to examine elliptic elements (cases \ref{coli} and \ref{colii} of \prettyref{thm:collins}). 
We follow the ideas of Anshel \cite{Anshel76conjugacy, Anshel76powers, Anshel76decision} in order to describe a $\AC(F_2)$ solution to the conjugacy problem in this case.

Let $v =a^k$, $w=b^\ell$ for some $a,b \in V(Y)$, $k,\ell \in \Z$. By \prettyref{thm:collins}, we know that 
$ v\sim_{F(\cG)} w$ if and only if there is some $z =a_0^{k_0} \y_1 a_1^{k_1} \cdots \y_n a_n^{k_n} \in \Pi(\cG, b,a)$ 
such that $z a^k z^{-1} \eqF w$ and 
\begin{align*}
	\y_i a_i^{k_i} \cdots \y_n\, a_n^{k_n} \cdot a^{k} \cdot a_n^{-k_n} \ov \y_n \cdots a_i^{-k_i} \ov \y_i &\in G_{\y_i}^{\y_i} &\text{for all } i.
\end{align*}
Since a conjugation with $a_i$ has no effect on elements of $G_{a_i} = \gen{a_i}$, we may assume that $z = \y_1 \cdots \y_n$ if $v$ and $w$ are conjugate.

Let $\cP = \oneset{p_1,\dots,p_m}$ as before be the set of prime divisors occurring in the $\alp_\y $ for $\y \in E(Y)$. Here and in what follows, we treat $-1$ as a prime number. Let
\begin{align}
	k &= r_k\cdot \prod_{i=1}^m p_i^{e_i(k)} , & \ell &= r_\ell\cdot\prod_{i=1}^m p_i^{e_i(\ell)} ,\label{eq:factorization}
\end{align}
such that $r_k,r_\ell > 0$ are not divisible by any $p \in \cP \setminus \oneset{-1}$.
The numbers $r_k,r_\ell$ and the exponents $e_i(k),e_i(\ell)$ can be computed in \TC as before by checking for all $p \in \cP$ and $e \leq \log \abs{k}$ (or, more precisely, for all $e$ at most the number of bits used to represent $k$) in parallel whether $p^e$ divides $k$ using Hesse's \TC circuit for \proc{Integer Division}, \prettyref{thm:divisionTC} (for $p=-1$, it has to be checked whether $k>0$)~-- and likewise for $\ell$. 
If $v \sim_{F(\cG)} w$, then $r_k=r_\ell$. Hence, all the information it remains to consider is given by the the vectors $(e_1(k), \dots , e_m(k)), (e_1(\ell), \dots , e_m(\ell)) \in \N^m$ and the vertices $a,b \in V(Y)$.
In order to code also the vertices as vectors, we consider vectors in $\N^m \times \N^{V(Y)}$ where a vertex $a$ is encoded by the unit vector $\vec u_a \in \N^{V(Y)}$ (which has a $1$ at position $a$ and $0$ otherwise).

Let us define an equivalence relation on $\N^m \times \N^{V(Y)}$ which reflects conjugacy in $F(\cG)$. For $\vec e = (e_1,\dots ,e_m,\vec u_a),\, \vec f = (f_1,\dots, f_m, \vec u_b) \in \N^m\times \N ^{V(Y)}$ with arbitrary $(e_1,\dots ,e_m)$, $(f_1,\dots, f_m)\in \N^m$ and $ a, b \in V(Y)$, we define $\vec e\sim \vec f$ if 
$$ a^{\prod_{i=1}^m p_i^{e_i}} \sim_{F(\cG)} b^{\:\!\prod_{i=1}^m p_i^{f_i}};$$
for $\vec e = (e_1,\dots ,e_m,\vec e'), \vec f = (f_1,\dots, f_m, \vec f') \in \N^m\times \N ^{V(Y)}$ with $\vec e'$ and $\vec f'$ not being zero nor a unit vector, we define $\vec e\sim \vec f$ regardless what the $e_i,f_i$ are. As an immediate consequence of this definition, we have

\begin{lemma}\label{lem:reductionCPtocomMon}
	Let $a,b \in V(Y)$, $k,\ell \in \Z$. Then $a^k\sim_{F(\cG)}b^\ell$ if and only if $$r_k = r_\ell \text{ and }(e_1(k), \dots , e_m(k),\vec u_a)\sim (e_1(\ell), \dots , e_m(\ell),\vec u_b).$$
\end{lemma}

The numbers $e_i(k), e_i(\ell)$ of \prettyref{eq:factorization} are bounded by a linear function in the input size. In particular, we have a \TC-many-one reduction from the question whether $a^k\sim_{F(\cG)}b^\ell$ to the question whether $(e_1(k), \dots , e_m(k),\vec u_a)\sim (e_1(\ell), \dots , e_m(\ell),\vec u_b) $ where the numbers $e_i(k),e_i(\ell)$ are represented in unary.
Thus, we aim for a $\AC(F_2)$ circuit to decide whether $\vec e\sim \vec f$ for vectors $\vec e, \vec f \in \N^m \times \N ^{V(Y)}$.
This can be achieved by using the following crucial observation, which is another immediate consequence of the definition of $\sim$.

\begin{lemma}\label{lem:congruence}
	If $\vec e \sim \vec f$, then also $\vec e + \vec g \sim \vec f + \vec g$ for all $g \in \N^m \times \N ^{V(Y)}$. In particular, $\sim$ defines a congruence on $\N^m \times \N ^{V(Y)}$.
\end{lemma}

Thus, $\rquot{\left(\N^m \times \N ^{V(Y)}\right)}{\sim}$ is a commutative monoid and it remains to solve the word problem of this monoid. Malcev \cite{Malcev58} and Emelichev \cite{Emelichev58} showed that the word problem for finitely generated commutative monoids is decidable~-- even if the congruence is part of the input. 

In \cite[Thm.\ II]{es69}, Eilenberg and Sch\"utzenberger showed that every congruence on $\N^M$ is a semilinear subset of $\N^M \times \N^M$ (this follows also from the results \cite{Taiclin68}, that congruences are definable by Presburger formulas, and \cite{gs66}, that Presburger definable sets are semilinear~-- for definition of all these notions we refer to the respective papers).
In \cite[Thm.\ 1]{IbarraJCR91}, Ibarra, Jiang, Chang, and
Ravikumar showed that membership in a fixed semilinear set can be decided in uniform $\NC^1$. As the word problem of $F_2$ is hard for uniform $\NC^1$ under \AC reductions \cite{Robinson93phd}, this means that for every fixed congruence ${\sim} \sse \N^M \times \N^M$, on input of $u,v \in \N^M$, it can be decided in $\AC(F_2)$ whether $u \sim v$.	

Thus, by \prref{lem:reductionCPtocomMon}, it can be decided in $\AC(F_2)$ whether $a^k \sim_{F(\cG)}b^\ell$ for $a,b \in V(Y)$, $k,\ell \in \Z$. 
Now, we can combine this result with \prettyref{cor:bsbred} (calculation of Britton-reduced $\cG$-factorizations) and \prettyref{prop:cycredconj} (solution to conjugacy in the hyperbolic case) and we obtain a proof of the main result on conjugacy, \prref{thm:cp}.

\begin{theorem}\label{thm:conjGBS}
	Let $G$ be a generalized Baumslag-Solitar group.
	Then the conjugacy problem of $G$ is in $\AC(F_2)$.
\end{theorem}

\subsection{The Uniform Conjugacy Problem}\label{sec:uniCP}
In \prettyref{sec:uniWP}, we have seen that the uniform version of the word problem for GBS groups was essentially as difficult as the word problem for a fixed GBS group. For conjugacy this picture changes dramatically. Like for the word problem in \prettyref{sec:uniWP}, the \emph{uniform conjugacy problem} for GBS groups receives as input a graph of groups $\cG$ consisting of a finite graph $Y$ and numbers $\alp_\y ,\bet_\y \in \Z\setminus\oneset{0}$ for $\y\in E(Y)$ and two $\cG$-factorizations $v,w \in \Delta^*$, where as before $\Del = E(Y) \cup \set{a^k}{a \in V(Y), k \in \Z}$. The question is whether $v \sim_{F(\cG)} w$ (what by \prref{lem:path} is equivalent to conjugacy in the fundamental group with respect to a base point).

In \cite{AnshelM83}, Anshel and McAloon considered a special (more difficult) variant of the uniform conjugacy problem; they showed that the so-called finite special equality problem for some GBS groups is decidable but not primitive recursive. However, they did not consider the uniform conjugacy problem.
By following the ideas for the non-uniform case (which themselves are based on Anshel's work \cite{Anshel76conjugacy, Anshel76powers, Anshel76decision}), we obtain a precise complexity estimate for the uniform conjugacy problem.

\begin{theorem}\label{thm:uniformconjugacy}
	The uniform conjugacy problem for GBS groups is \EXPSPACE-complete~-- even if the numbers $\alp_\y , \bet_\y $ are given in unary. 
\end{theorem}
This concludes the proof of \prref{thm:uni}. The proof of \prettyref{thm:uniformconjugacy}  is an application of the next theorem by Cardoza, Lipton and Meyer \cite{CardozaLM76} resp.\ Mayr and Meyer \cite{mm82}. 
\begin{theorem}[\cite{CardozaLM76,mm82}]\label{thm:uniformwpcommon}
	The uniform word problem for finitely presented commutative semigroups is \EXPSPACE-complete.
\end{theorem}

\begin{proof}[Proof of \prettyref{thm:uniformconjugacy}]
	For the hardness part, we give a \L reduction from the uniform word problem of f.\,g.\ commutative semigroups to the uniform conjugacy problem for GBS groups. \Wlog we only consider commutative monoids.
	Let $m \in \N$, $e,f \in \N^m, (r_i,s_i)_{i\in\oneset{1\dots n}}$ with $r_i,s_i \in \N^m$ be some instance for the uniform word problem of commutative monoids (\ie the question is whether $e\sim f$ for the smallest congruence $\sim$ satisfying $r_i \sim s_i$ for all $i$). 
	
	We construct an instance for the uniform conjugacy problem as follows: The graph $Y$ consists of a single vertex $a$; for all $i \in \oneset{ 1, \dots, n}$ there is a pair of edges $\y_i,\ov \y_i \in E(Y)$. Let $\cP=\oneset{p_1,\dots, p_m}$ be the set of the first $m$ prime numbers. The numbers $p_j$ can be computed in \L since each of them requires a logarithmic (in $m$) number of bits, only (by the prime number theorem there are enough primes).
	Now, for every relator $(r_i,s_i)$, we define $\alp_{\y_i}=\prod_{j=1}^m p_j^{(r_i)_j}$ and $\bet_{\y_i}=\prod_{j=1}^m p_j^{(s_i)_j}$, where $(r_i)_j$ denotes the $j$th component of the vector $r_i$, and  $k=\prod_{j=1}^m p_j^{e_j}$ and $\ell=\prod_{j=1}^m p_j^{f_j}$. According to the proof in \cite{mm82}, we may assume that all the vectors $e$, $f$, $r_i$ and $s_i$ (for all $i$) have at most four non-zero entries and these non-zero entries are at most $2$. Thus, the results $k$, $\ell$, $\alp_{\y_i}$, and $\bet_{\y_i}$ are bounded polynomially in the input length and they can be written down in unary on the output tape. In particular, the products can be computed in \L.
	Now we have $a^k \sim_{F(\cG)} a^\ell$ if and only if $e \sim f$.
	
	It remains to show that the uniform conjugacy problem is in \EXPSPACE. The two input words for an instance of the uniform conjugacy problem for GBS groups can be cyclically Britton-reduced as in \prref{cor:bsbred}. Note, however, that the linear bound on the size of the cyclically Britton-reduced words does not hold anymore. Still the size remains bounded polynomially.
	
	The algorithm of \prref{prop:cycredconj} can be executed in polynomial time even if the graph of groups is part of the input. This gives a polynomial time bound for hyperbolic elements. However, we do not know a better bound as the proof of \prref{prop:cycredconj} involves a computation of greatest common divisors (or prime factorizations) of the numbers $\alp_\y , \bet_\y $. 
	For elliptic elements, by \prettyref{lem:reductionCPtocomMon}, we obtain an instance of the uniform word problem of commutative semigroups, which is in \EXPSPACE.
\end{proof}

\paragraph{Acknowledgments.}
Many thanks goes to my thesis advisor, Volker Diekert, as well as to Jonathan Kausch with whom I had so many inspiring conversations about complexity and Baumslag-Solitar groups.
Special thanks to the anonymous referee for many helpful comments and, in particular, for pointing out that the word problem for a fixed commutative monoid actually can be solved in $\NC^1$~-- thus, providing the last piece for the solution of the conjugacy problem in $\AC(F_2)$ and not only in \L.

\bibliographystyle{abbrv}

\end{document}